\documentclass[conference]{IEEEtran}
\IEEEoverridecommandlockouts
\usepackage{amsmath,amssymb,amsfonts}
\usepackage{algorithmic}
\usepackage{graphicx}
\usepackage{textcomp}
\usepackage{xcolor}
\usepackage{multirow}
\usepackage{graphicx}
\usepackage{amsmath}

\usepackage{amsthm}
\usepackage[backend=biber,style=ieee]{biblatex}
\usepackage{amssymb}
\usepackage{algorithmic}
\usepackage{array}
\usepackage{stfloats}
\usepackage{stmaryrd}
\usepackage{float}
\usepackage{subcaption}
\usepackage{tikz}
\usepackage[normalem]{ulem}
\usetikzlibrary{shapes.gates.logic.US,trees,positioning,arrows}
\usetikzlibrary{calc}
\usepackage{wrapfig}
\usepackage{cutwin}
\usepackage{pgfplots}
\usepgfplotslibrary{statistics}
\usepackage{enumitem}
\usepackage{filecontents}
\usepackage{thmtools}
\usepackage{url}
\usepackage{thm-restate}
\usepackage{hyperref}
\newcommand\recht\operatorname
\newcommand{\BB}{\mathbb{B}}
\newcommand{\type}{\gamma}
\newcommand{\BAS}{\mathrm{BAS}}
\newcommand{\OR}{\mathrm{OR}}
\newcommand{\AND}{\mathrm{AND}}
\newcommand{\R}[1]{\recht{R}_{#1}}
\newcommand{\ch}{\recht{Ch}}
\newcommand{\bx}{\mathbf{x}}

\newcommand{\by}{\mathbf{y}}
\newcommand{\FA}{\mathcal{A}}
\newcommand{\cc}{\recht{c}}
\newcommand{\hc}{\hat{\cc}}
\newcommand{\dd}{\recht{d}}
\newcommand{\hd}{\hat{\dd}}
\newcommand{\umin}{\underline{\min}}
\newcommand{\Rnn}{\mathbb{R}_{\geq 0}}
\newcommand{\DD}{\Rnn^2}
\newcommand{\poD}{\sqsubseteq}
\newcommand{\spoD}{\sqsubset}
\newcommand{\E}{\binom{\hc}{\hd}}
\newcommand{\PF}{\recht{PF}(T)}
\newcommand{\dopt}{d_{\recht{opt}}}
\newcommand{\copt}{c_{\recht{opt}}}
\newcommand{\DB}{\mathtt{DTrip}}
\newcommand{\C}[1]{\mathcal{C}_{#1}^{\recht{D}}}
\newcommand{\CU}{\C{U}}
\newcommand{\m}[1]{\umin_{#1}}
\newcommand{\mU}{\m{U}}
\newcommand{\Pow}{\mathcal{P}}
\newcommand{\pp}{\recht{p}}
\newcommand{\Yx}{Y_{\bx}}
\newcommand{\hdE}{\hd_{\recht{E}}}
\newcommand{\Ex}{\mathbb{E}}
\newcommand{\EE}{\binom{\hc}{\hdE}}
\newcommand{\Prob}{\mathbb{P}}
\newcommand{\kk}{\recht{PS}}
\newcommand{\DP}{\mathtt{PTrip}}
\newcommand{\PFE}{\recht{EPF}(T)}

\newcommand{\CpU}{\mathcal{C}_{U}^{\recht{P}}}

\newcommand{\tDF}[1]{\widetilde{\mathtt{FTrip}}_v}
\newcommand{\DF}[1]{\mathtt{FTrip}_v}

\newcommand\rnd\mathsf
\newcommand{\struc}{\recht{S}}
\newcommand{\vvec}[1]{\left(\begin{smallmatrix}#1\end{smallmatrix}\right)}

\newcommand{\hlbox}[1]{%
  \smallskip\begin{center}
  \fboxrule1pt\fboxsep3pt\fcolorbox{black!45}{black!8}{%
  \begin{minipage}{.96\linewidth}#1\end{minipage}}
  \end{center}\smallskip}

\allowdisplaybreaks
\newtheorem{definition}{Definition}
\newtheorem{example}{Example}
\newtheorem{theorem}{Theorem}
\newtheorem{lemma}{Lemma}
\newtheorem*{myproblem}{Problem}

\def\BibTeX{{\rm B\kern-.05em{\sc i\kern-.025em b}\kern-.08em
    T\kern-.1667em\lower.7ex\hbox{E}\kern-.125emX}}
\begin{document}

\title{Cost-damage analysis of attack trees\thanks{This research has been partially funded   by ERC Consolidator grant 864075 CAESAR and the European Union’s Horizon 2020 research and innovation programme under the Marie Skłodowska-Curie grant agreement No. 101008233.}}
\author{\IEEEauthorblockN{Milan Lopuhaä-Zwakenberg}
\IEEEauthorblockA{\textit{University of Twente}\\
m.a.lopuhaa@utwente.nl}
\and
\IEEEauthorblockN{Mariëlle Stoelinga}
\IEEEauthorblockA{\textit{University of Twente} \& \textit{Radboud University}\\
m.i.a.stoelinga@utwente.nl}
}

\maketitle

\thispagestyle{plain}
\pagestyle{plain}

\begin{abstract}
Attack trees (ATs) are a widely deployed modelling technique to categorize potential attacks on a system. An attacker of such a system aims at doing as much damage as possible, but might be limited by a cost budget. The maximum possible damage for a given cost budget is an important security metric of a system. In this paper, we find 
the maximum damage given a cost budget by modelling this problem with ATs, both in deterministic and probabilistic settings. We show that the general problem is NP-complete, and provide heuristics to solve it. For general ATs these are based on integer linear programming.  However when the AT is tree-structured, then one can instead use a faster bottom-up approach. We also extend these methods to other problems related to the cost-damage tradeoff, such as the cost-damage Pareto front.
\end{abstract}


\section{Introduction}

\begin{wrapfigure}[26]{r}{4cm}
\centering
\includegraphics[width=4cm]{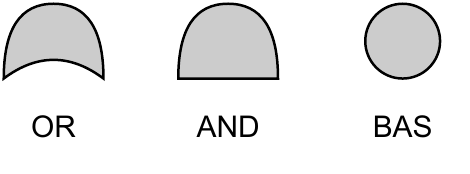}

\includegraphics[width=3.5cm]{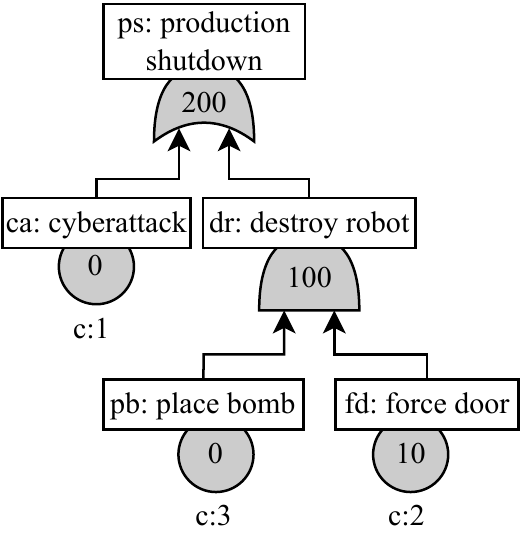}

\caption{Attack tree for a factory. Production can be stopped by a cyberattack or by destroying the production robot, for which an attacker forces their way inside and places a bomb. Damage values (in 1000 USD) are inscribed in the nodes, and cost values are below the BASs.
}
\label{fig:bank1}
\end{wrapfigure}

\noindent \textbf{Attack trees.} Attack trees (ATs) are a prominent methodology in security analysis. They aid
security specialists in identifying, analyzing and prioritizing
 (cyber)risks.
ATs are included in several popular system engineering frameworks, e.g., \emph{UMLsec} \cite{RA15} and \emph{SysMLsec} \cite{AR13}, and are supported by industrial tools such as Isograph's \emph{AttackTree} \cite{IsographAT}. ATs have been used in many scenarios, such as military information infrastructure \cite{Surdu2003}, electronic voting \cite{Yasinisac2011}, and IoT insider threats \cite{Kammueller2016}. Their popularity is owed to their simplicity, which allows for a range of applications, and their analyzability.

An AT is an hierarchical diagram that describes a system's vulnerabilities to attacks.  Despite the name, an AT is a rooted directed acyclic graph (DAG). Its root represents the adversary's goal, while leaves represent basic attack steps (BASs) undertaken by the adversary. Other nodes represent intermediate attack goals and are labeled with an OR-gate or AND-gate, determining how its activation depends on that of its children. An example is given in Fig.~\ref{fig:bank1}.

\noindent \textbf{Quantitative analysis.} Besides describing possible attacks on a system, ATs can also be used to analyze quantitative information about such attacks. Many \emph{attack metrics} exist, such as the damage, required cost, or required skill of an attack. 
Such metrics are key performance indicators that formalize a system's security performance.

These metrics do not exist in isolation, and their interplay is  important for quantitative security analysis. For instance, one attack may be cheaper than another, but require more time, or a more skilled attacker. Therefore, it is essential to understand the tradeoff between different security metrics. To understand and quantify such tradeoffs, one considers the \emph{Pareto front} of multiple metrics \cite{fila2019efficient}, which includes all attacks that are not dominated by another attack in all metrics. For instance, in Fig.~\ref{fig:bank1} the attack $\{\mathtt{ca}\}$ does damage 200 for cost 1, which is preferable over $\{\mathtt{fd}\}$ which does 10 damage for cost 2.

\noindent \textbf{Cost-damage analysis.} In this paper we consider the interplay between two important attack metrics: \emph{attack cost} \cite{AGKS15}, describing an attacker's budget in, e.g., money or time;  and \emph{attack damage} \cite{saini2008threat}, representing the damage done to the system, e.g., in terms of monetary value. The larger the cost budget available to an attacker, the more damaging an attack can be. While damage is the most relevant metric to the system owner, knowing the cost of an attack helps them understand the likelihood of such an attack. This fits within the perspective that likelihood and impact both play an important role in risk analysis \cite{hopkin2018fundamentals}.
For a comprehensive risk assessment of a system's security, it is therefore paramount to solve the following problems: 

\hlbox{{\bf Problem statement.}
Given an attack tree $T$, solve the following problems:
\begin{enumerate}[topsep=0pt,leftmargin=32pt]
    \item[DgC)] Find the most \textbf{D}amaging attack \textbf{g}iven a \textbf{C}ost budget.
    \item[CgD)] Find the \textbf{C}heapest attack \textbf{g}iven a \textbf{D}amage threshold.
    \item[CDPF)]  Find the \textbf{C}ost-\textbf{D}amage \textbf{P}areto \textbf{F}ront.
\end{enumerate}
}

Existing approaches to calculating the Pareto front of multiple AT metrics \cite{kumar2015quantitative,fila2019efficient,BS21} cannot be applied to cost-damage problems for two reasons: First, existing methods assume that only BASs are assigned metric values. For damage, this assumption is not realistic, as the internal nodes often represent disabled subsystems, which also have an associated damage value. For instance, in Fig. \ref{fig:bank1}, the attack $\{\mathtt{ca}\}$ and $\{\mathtt{pb},\mathtt{fd}\}$ both shut down production, but the latter does so by destroying the production robot, leading to greater monetary loss.
Second, existing methods only consider \emph{successful attacks}, i.e., attacks that activate the top node of the AT. In the case of cost-damage analysis, however, attacks not reaching the top node can still do quite some damage on intermediate nodes, and should be considered in the analysis. For instance, an attacker can try to rob an ATM by forcing it with explosives. Even if the attacker fails in stealing the money, the explosives still cause significant damage to the ATM owner. Thus existing work cannot solve cost-damage problems in the generality required to model realistic scenarios.
For these reasons new approaches and algorithms need to be developed.

\smallskip
\noindent \textbf{Approach.} This paper 
introduces 
three novel methods to solve the problems stated above. We first consider a deterministic setting, where BASs always succeed. We then consider a probabilistic setting, where BASs may fail with a given probability. 

\emph{NP-completeness:} 
We first prove two important negative results, showing that even the simplest cost-damage problems do not have `easy' solutions. Cost-damage problems are similar to binary knapsack problems \cite{dudzinski_exact_1987}; we use this to prove that even the simplest type of cost-damage analysis is NP-complete.
Unfortunately, this similarity cannot be exploited to apply heuristics for knapsack problems or their many extensions \cite{Gallo1980,forrester_strengthening_2022,SVIRIDENKO200441} to cost-damage problems:  All 
extensions assume properties of the damage function 
(i.e., the function assigning a damage value to each attack)
that are not met in our setting.
In fact, we prove that the damage function can be any nondecreasing function. This highlights the need for the completely new methods for cost-damage analysis in ATs.

As common, our algorithms distinguish between tree- and DAG-shaped attack trees. Further, we
consider deterministic versus probabilistic failure behaviour in the leaves. 

\emph{Bottom-up algorithm for treelike ATs:} 
Existing approaches to the Pareto front of two metrics work bottom-up, discarding non-optimal attacks at every node \cite{fila2019efficient}. This does not work for damage, as intermediate nodes also carry damage values. Hence attacks that are non-optimal at a certain node may do more damage at a higher node, becoming optimal there.

To solve problem CDPF above, we describe a new bottom-up method for finding the Pareto front in both the deterministic and probabilistic setting.
The key insight is to perform
a bottom-up Pareto analysis
in an {\em extended cost-damage domain}, by adding a dimension
for the current top node's activation (or activation probability in the probabilistic setting); this dimension signifies an attack's `potential' to do more damage at higher nodes.
As shown in our experiments, these bottom-up methods drastically reduce computation time from multiple hours to less than 0.1 second.

For the single-objective problems DgC and CgD we cannot do a `simpler' bottom-up approach in which only the optimal attack is propagated, as one needs the overview of the full AT to decide which attack is optimal. Instead, we still need to propagate (part of) the Pareto front, and we gain our solution for DgC and CgD from minor adaptations to the CDPF approach.

\emph{Integer linear programming for DAG-like ATs:} 
It is well-known \cite{Kordy2018,BS21} that 
bottom-up algorithms do not work for DAG-like ATs:  since nodes may have multiple parents,  their cost/damage being counted twice. We introduce a novel method for the deterministic setting {by translating cost-damage problems into the \emph{bi-objective integer linear programming} (BILP) framework \cite{ozlen2009multi}; we can then apply existing BILP solvers to solve them \cite{gurobi}}. This {translation} is nontrivial, as damage is a nonlinear function of the adversary's attack, as we will show in Section \ref{sec:knapsack}.
The key insights behind our algorithm are that (1) damage is linear in terms of the \emph{structure function} that describes which AT nodes are reached by an attack and (2) the constraints defining the structure function can be phrased as linear constraints.

We use existing biobjective methods and solvers to solve CDPF \cite{stidsen2014branch}, and single-objective solvers to solve DgC and CgD \cite{yalmip}. This does not extend to the probabilistic setting, where equations become nonlinear; we leave the analysis of probabilistic DAG-like ATs as an open problem. 

Finally, in experiments we show our methods can be used for risk analysis by applying them to two systems: a wireless sensor {device} tracking wildlife in a giant panda reservation, and a data server in a network behind a firewall. The ATs of these systems are taken from the literature \cite{jiang2012attack,dewri2012optimal}. We use the cost-damage Pareto front to assess the vulnerabilities of these systems. Furthermore, we also measure the computing time {in the case studies and on 500 random ATs: both bottom-up and BILP methods vastly outperform the existing enumerative approach.}
This shows that our methods present an enormous speedup compared to the status quo.

\begin{table}
\centering
\begin{tabular}{c|cc}
     & Tree & DAG \\
     \hline
Deterministic     & bottom-up (Theorem \ref{thm:CDPF-tree}) & BILP (Theorem \ref{thm:bilp}) \\
Probabilitistic & bottom-up (Theorem \ref{thm:CEDPF-tree}) & \emph{open problem}\\
\end{tabular}
\caption{Overview of this paper's algorithmic contributions.}
\end{table}

\hlbox{{\bf Contributions.} Summarized, our contributions are: 
    \begin{enumerate}[topsep=0pt]
    \item A formal definition of cost-damage problems in ATs;
    \item A proof that these problems are NP-complete (Sec.~\ref{sec:knapsack});
    \item A proof that cost-damage problems cannot be reduced to common extensions of the binary knapsack problem; (Sec.~\ref{sec:knapsack});
    \item A bottom-up method to solve the deterministic and probabilistic cost-damage problems for treelike attack trees (Sec.~\ref{sec:treedet} \& \ref{sec:treeprob});
    \item An integer linear programming-based method to solve the deterministic cost-damage problems for DAG-like attack trees (Sec.~\ref{sec:dagdet}).
    \item An experimental evaluation of the above methods on two realistic cases from the literature (Sec.~\ref{sec:exp}).
\end{enumerate}
}

The Matlab code for the experiments can be found at \cite{code}.

\section{Related work} \label{sec:rel}

In the literature, there are multiple approaches to decorating an AT with cost and damage values. Existing work concerning damage (also called \emph{impact}) on ATs can be divided into three categories: works in which only BASs have a damage attribute {\cite{saini2008threat,bobbio2013methodology,10.1007/978-3-319-11599-3_12,kumar2015quantitative}}, works in which only the root node has a damage attribute \cite{jurgenson2008computing}, and works in which every node can have a damage attribute \cite{ingoldsby2005fundamentals}. In the same manner, in some works intermediate nodes are allowed to have an associated cost \cite{andre2021parametric,kumar2015quantitative}, while in others only BASs have costs {\cite{mauw2005foundations,kumar2015quantitative,BS21}.}
In this paper, every node has a damage attribute, while only BASs have a cost attribute. We choose this because it is the simplest model for the most expressivity; as we will show in Section \ref{sec:cd}, cost values on internal nodes can be modeled by adding dummy BASs, but damage values cannot.

{Most of the work listed above only considers one metric at a time. For instance, in \cite{bobbio2013methodology} binary decision diagrams (BDDs) are used to calculate both the minimal cost of a succesful attack and the maximal damage, but the tradeoff between the two metrics is not investigated. Other methods for calculating single metrics include bottom-up methods for treelike ATs \cite{BS21} and priced-timed automata \cite{andre2021parametric}. Of the works that consider cost-damage tradeoffs, some focuse on modeling rather than algorithms \cite{ingoldsby2005fundamentals,saini2008threat}. One approach to the Pareto front is via priced-timed automata \cite{kumar2015quantitative}; however, we cannot directly apply this to our setting as in that work only BASs have a damage attribute. In \cite{jurgenson2008computing}, cost and damage are used to define a single attack parameter \emph{outcome}, which is optimized heuristically.

Other works on ATs consider the Pareto front between two generic metrics. A bottom-up method for calculating Pareto fronts for treelike ATs, and under some additional assumption for DAG-like ATs, is given in \cite{fila2019efficient}. Furthermore, a BDD-based approach for DAG-like ATs is developed in \cite{BS21}. However, damage does not satisfy the conditions for either of these two approaches, and these cannot be used for our CgD, DgC and CDPF problems. Overall, we can conclude that none of the existing literature is able to solve cost-damage problem in the general model discussed in this paper.
}


Another approach to multi-objective optimization is to approximate the Pareto front, for example using genetic algorithms \cite{deb2002fast,ali2020quality}. {This has also been applied to ATs with cost \cite{10.1007/978-3-319-11599-3_12}.} While such an approach would be interesting for cost-damage ATs, in this paper we instead focus on provably optimal solutions, corresponding to provable security guarantees.

\section{Preliminaries} \label{sec:prelim}

\begin{table}[t]
\centering
\begin{tabular}{ccc}
Notation & Explanation & page  \\ \hline
$\BB$ & \{0,1\} & \pageref{page:b} \\
$T = (N,E)$ & Attack Tree & \pageref{def:at} \\
$B$ & BASs of $T$ & \pageref{page:bas} \\
$\gamma(v)$ & Type of node $v$ & \pageref{def:at} \\
$\ch(v)$ & Children of node $v$ & \pageref{page:ch} \\ 
$(\FA,\preceq)$ & Poset of attacks & \pageref{def:att} \\
$\struc(\bx,v)$ & Structure function of $T$ & \pageref{def:sf} \\
$\cc(v)$ & Cost of BAS $v$ & \pageref{def:cd} \\
$\dd(v)$ & Damage of node $v$ & \pageref{def:cd} \\
$\hc(\bx)$ & Cost of attack $\bx$ & \pageref{def:cd}\\
$\hd(\bx)$ & Damage of attack $\bx$ & \pageref{def:cd}\\
$\umin{} X$ & Set of minima of $X$ & \pageref{page:umin} \\
$(\DD,\sqsubseteq)$ & Poset of attribute pairs & \pageref{page:pair} \\
$\E$ & Attribution map & \pageref{page:attr} \\
$\PF$& Pareto-front of $T$ & \pageref{page:CDPF} \\
CDPF & Cost-damage Pareto front & \pageref{page:CDPF} \\
DgC & Maximal damage given cost & \pageref{page:DgC}\\
CgD & Minimal cost given damage & \pageref{page:CgD}\\
$(\DB,\sqsubseteq)$ & Deterministic attribute triples & \pageref{page:dtrip}\\
$\mU$ & Cost-restricted $\umin$ & \pageref{page:minu1},\pageref{page:umin2}\\
$\CU(v)$ & Incomplete deterministic PF at $v$ & \pageref{page:cu}\\
$\pp(v)$ & Probability of BAS $v$ & \pageref{def:cdp}\\
$Y_{\bx}$ & Actualized attack & \pageref{def:de}\\
$\hdE(\bx)$ & Expected damage of attack $\bx$ & \pageref{def:de} \\
CEDPF & Cost-expected damage Pareto front & \pageref{prob:cedpf}\\
EDgC & expected damage given cost & \pageref{prob:edgc}\\
CgED & cost given expected damage & \pageref{prob:cged}\\
$\kk(\bx,v)$ & Probabilistic structure function & \pageref{page:ps}\\
$(\DP,\sqsubseteq)$ & Probabilistic attribute triples & \pageref{eq:ptrip}\\
$\CpU(v)$ & Incomplete probabilistic PF at $v$ & \pageref{page:cpu}
\end{tabular}
\caption{Notation used in this paper.}
\label{tab:nota}
\end{table}

Let $\BB$ \label{page:b} be the set $\{0,1\}$, with logical operators ${\wedge},{\vee}$.

\begin{definition} \label{def:at}
An \emph{attack tree} is a rooted directed acyclic graph $T = (N,E)$ where each node $v \in N$ has a type $\type(v) \in \{\BAS,\OR,\AND\}$, such that $\type(v) = \BAS$ if and only if $v$ is a leaf.
\end{definition}

Contrary to terminology an AT is not necessarily a tree. When the DAG $T$ is actually a tree, it is called \emph{treelike}; the general case is referred to as \emph{DAG-like}. The root of $T$ is denoted $\R{T}$. For a node $v$ we denote its set of children by $\ch(v) = \{w \mid (v,w) \in E\}$; \label{page:ch} we also say that $v$ is an \emph{ancestor} of $w$, and $w$ a \emph{descendant} of $v$, if there is a path $v \rightarrow w$ in $T$.
When $\ch(v) = \{v_1,\ldots,v_n\}$, we write $v = \OR(v_1,\ldots,v_n)$ or $v = \AND(v_1,\ldots,v_n)$ depending on $\type(v)$. The set of BASs a.k.a. leaves is denoted by $B$. \label{page:bas}
For instance, in the AT $T$ from Fig. \ref{fig:bank1} one has $B = \{\mathtt{ca}, \mathtt{pb},\mathtt{fd}\}$, $\mathtt{dr} = \AND(\mathtt{pb},\mathtt{fd})$, and $\R{T} = \mathtt{ps} = \OR(\mathtt{ca},\mathtt{dr})$. 
Note that $T$ is treelike.

An attacker performs an attack by activating a chosen set of BASs, represented by a \emph{status vector} $\mathbf{x} \in \BB^B$; the status $x_v$ of a BAS $v$ equals 1 if $v$ is activated, and $0$ if it is not. Such a status vector can also be regarded as a subset of $B$. Transposing the partial order $\subseteq$ to status vectors yields a partial order $\preceq$.

\begin{definition} \label{def:att}
An \emph{attack} on $T$ is a vector $\bx \in \BB^B$; we let $\FA = \BB^B$ be the set of all attacks. This has a partial order $\preceq$ given by $\bx \preceq \by$ iff $x_v \leq y_v$ for all $v \in B$.
\end{definition}

An attack propagates upwards from the BASs. A node is reached by an attack depending on its type $\OR$ or $\AND$, and whether any/all of its children are reached by the attack. This idea is formalized by the structure function $\struc$.
Given an attack vector $\bx$, 
and a node $v$,
$\struc(\bx,v)$
indicates whether 
 $v$ is reached by  $\bx$, i.e., if $\struc(\bx,v) = 1$.

\begin{definition} \label{def:sf}
The \emph{structure function} $\struc\colon \FA \times N \rightarrow \BB$ of $T$ is defined recursively:
\begin{equation*}
\struc(\bx,v) = \begin{cases}
x_v & \textrm{ if $\type(v) = \BAS$},\\
\bigvee_{v' \in \ch(v)} \struc(\bx,v') & \textrm{ if $\type(v) = \OR$},\\
\bigwedge_{v' \in \ch(v)} \struc(\bx,v') & \textrm{ if $\type(v) = \AND$}.
\end{cases}
\end{equation*}
\end{definition}

\section{Deterministic cost-damage problems for ATs} \label{sec:cd}

In this section we formulate this paper's problem; solutions are presented  in Sections \ref{sec:treedet} and \ref{sec:dagdet}. This section deals with a deterministic setting, where a BAS's success is guaranteed; its probabilistic equivalent is presented in Section \ref{sec:prob}.

 The attacker's goal is to disrupt the system as much as possible, which is measured by a \emph{damage} value representing financial cost, downtime, etc. Each node $v$ has a damage value $\dd(v)$, and an attack's total damage $\hd(\bx)$ is the sum of the damage value of all nodes reached by $\bx$. 
 At the same time, an attacker may have only limited resources. Each BAS $v$ has a \emph{cost} value $\cc(v)$ representing e.g. the money, time or resources the attacker has to spend to activate it. The total cost $\hc(\bx)$ of an attack is the sum of the costs of the activated BASs.

\begin{definition} \label{def:cd}
A \emph{cd-AT} is a triple $(T,\cc,\dd)$ of an AT $T$ and maps $\cc\colon B\rightarrow \mathbb{R}_{\geq 0}$ and $\dd\colon N \rightarrow \mathbb{R}_{\geq 0}$. 
Define the total cost and damage functions $\hc,\hd\colon \FA \rightarrow \Rnn$ by
\begin{align*}
\hc(\bx) &= \sum_{v \in B} x_v\cc(v), & \hd(\bx) &= \sum_{v \in N} \struc(\bx,v)\dd(v).
\end{align*}
\end{definition}

As opposed to other works in quantitative analysis on ATs \cite{fila2019efficient,BS21}, we do not only consider so-called \emph{{successful}} attacks, i.e., $\bx$ for which $\struc(\bx,\R{T}) = 1$. The reason is that in our model damage can be done at any level, not just at the top node. It is therefore important to know the damaging capabilities of an attacker, even when that attacker's limited resources mean that they cannot damage the top node. {Furthermore, an attacker may try different avenues towards success, and while a given path may be discarded without reaching the top node, side effects may remain. We therefore assign damage values not only to the top node, but also to internal nodes.}

\begin{example} \label{ex:bank2}
Consider the AT $T$ from Fig. \ref{fig:bank1}, repeated below, and its cost and damage functions. Then the functions $\hc$ and $\hd$ are calculated as in the following table.
\begin{center}
\begin{tabular}{c|cccccccc}
$x_{\mathtt{ca}}$ & 0 & 0 & 0 & 0 & 1 & 1 & 1 & 1 \\
$x_{\mathtt{pb}}$ & 0 & 0 & 1 & 1 & 0 & 0 & 1 & 1 \\
$x_{\mathtt{fd}}$ & 0 & 1 & 0 & 1 & 0 & 1 & 0 & 1 \\
\hline
$\hc(\bx)$ & 0 & 2 & 3 & 5 & 1 & 3 & 4 & 6 \\
$\hd(\bx)$ & 0 & 10 & 0 & 310 & 200 & 210 & 200 & 310\\
\end{tabular}
\end{center}
\end{example}

\begin{wrapfigure}[9]{r}{4cm}
\centering
\vspace{-1em}
\includegraphics[width=3.5cm]{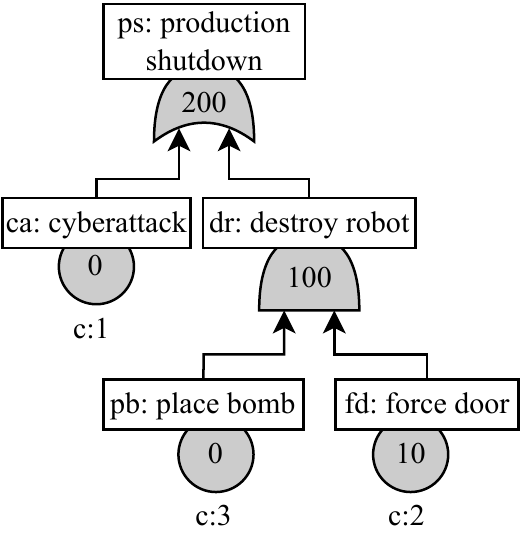}
\end{wrapfigure}

Some works also assign cost values to internal nodes \cite{andre2021parametric,kumar2015quantitative}, the interpretation being that an internal node is only activated if enough of its children are activated and its cost is paid. However, this can be simulated by adding a dummy BAS which holds the associated cost, as in Fig. \ref{fig:dummy}. However, the same cannot be done for damage: moving the damage to the dummy BAS leads to a situation where \emph{only} the dummy needs to be activated to do the damage. For full expressivity we thus allow internal nodes to have damage values, but not cost values. 

\begin{figure}
\centering
\begin{subfigure}{0.25\linewidth}
\centering
\includegraphics[width=2.2cm]{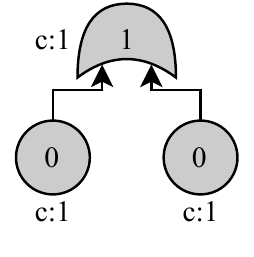}
\end{subfigure}
\begin{subfigure}{0.35\linewidth}
\centering
\includegraphics[width=3cm]{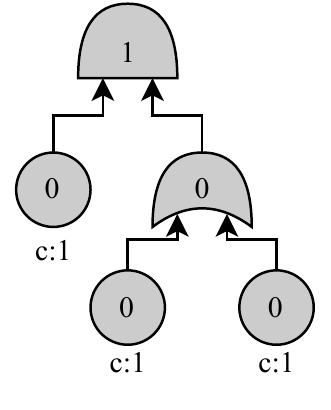}
\end{subfigure}
\begin{subfigure}{0.35\linewidth}
\centering
\includegraphics[width=3cm]{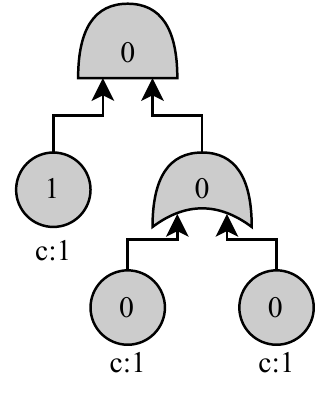}
\end{subfigure}
\caption{An example showing that damage values on internal nodes are necessary, but cost values on internal nodes are not. The cost value on the internal node in the AT on the left is replaced by a dummy BAS in the middle AT, which is equivalent: both ATs require cost 2 to perform 1 damage. In the right AT, the damage is also moved to the dummy BAS, but the result is not equivalent: 1 cost already yields 1 damage.} \label{fig:dummy}
\end{figure}

\subsection{Cost damage problems}

\setlength{\textfloatsep}{10pt}

In ATs, there is a tradeoff between resource utilization and damage: the higher the cost budget an attacker has at their disposal, the more damage they may cause. This tradeoff can be analyzed via the \emph{Pareto front}: the cost and damage values of all attacks that are not dominated by other attacks, where $\bx$ dominates $\by$ if $\bx$ is cheaper than $\by$ while doing more damage. An attack $\bx$ in the Pareto front is called \emph{Pareto optimal}, and it is the most damaging attack if the attacker cannot exceed cost $\hc(\bx)$. Thus the Pareto front gives a full overview of the system's vulnerability to any attacker.

For a general poset $(X,\preceq)$, we define its set of minimal elements as \label{page:umin}
\begin{equation*} 
\umin_{\preceq} \ X = \{x \in X \mid \forall x' \in X. x' \not \prec x\}.
\end{equation*}
We drop the subscript $\preceq$ if it is clear from the context.
We consider the domain of \emph{attribute pairs}, i.e., the set $\DD$ with a partial order $\poD$ given by $(a,a') \poD (b,b')$ if and only if $a \leq b \textrm{ and } a' \geq b'$.\label{page:pair}
For a cd-AT $(T,\cc,\dd)$, we define the evaluation map $\E\colon \FA \rightarrow \DD$ by $
\E(\bx) = \binom{\hc(x)}{\hd(x)}$\label{page:attr} (we represent elements of $\DD$ as column vectors). Note that $\bx$ dominates $\by$ if and only if $\E(\bx) \spoD \E(\by)$.

The aim of this paper is to find the cost-damage Pareto front, as well as two related single-objective problems. Mathematically, these are formulated as follows:

\hlbox{
\textbf{Problems.}
Given a cd-AT $(T,\cc,\dd)$, solve the following problems:
\begin{enumerate}[topsep=0pt,leftmargin=30pt]
    \item[\textbf{CDPF}] Cost-damage Pareto front: find $\umin_{\poD} \E(\FA) \subseteq \DD$.\label{page:CDPF}
    \item[\textbf{DgC}] Maximal damage given cost constraint: Given $U \in \mathbb{R}_{\geq 0}$, find $\dopt = \max_{\bx\colon \hc(\bx) \leq U} \hd(\bx)$.\label{page:DgC}
    \item[\textbf{CgD}]  Minimal cost given damage constraint: $L \in \mathbb{R}_{\geq 0}$, find $\copt = \min_{\bx\colon \hd(\bx) \geq L} \hc(\bx)$.\label{page:CgD}
\end{enumerate}
}

From CDPF one can {solve} DgC and CgD via
\begin{align}
\dopt &= \max \{d \in \Rnn \mid \exists c \in [0,U]. \vvec{c\\d} \in \PF\},\label{eq:DgC}\\
\copt &= \min \{c \in \Rnn \mid \exists d \in \mathbb{R}_{\geq L}. \vvec{c\\d} \in \PF\}. \label{eq:CgD}
\end{align}

{These problems are relevant in security analysis: DgC can be used to determine the damaging capabilities of different attacker profiles \cite{10.1007/978-3-319-11599-3_12,kumar2015quantitative}.
CDPF can be used to give an overview over all attacker profiles.
For a security operations center monitoring a network, a cost-damage analysis (with cost measured in time) provides insight in whether the response time is sufficient to stop damaging attacks.}

\begin{example} \label{ex:bank3}
In Example \ref{ex:bank2}, $\E(\FA)$ is given by the lower two rows of the table.
A number of these attacks are not Pareto optimal: we have $\vvec{1\\200} \sqsubset \vvec{2\\10},\vvec{3\\0},\vvec{4\\200}$, and furthermore $\vvec{5\\310} \sqsubset \vvec{6\\310}$. It follows that (see Fig. \ref{fig:PF}):
 \begin{equation}
 \PF = \left\{\vvec{0\\0},\vvec{1\\200},\vvec{3\\210},\vvec{5\\310}\right\}.\label{eq:bankpf}
 \end{equation}
 From this we find, for instance that the solution to DgC for $U=2$ is given by $d_{\recht{opt}} = 200$.
\end{example}

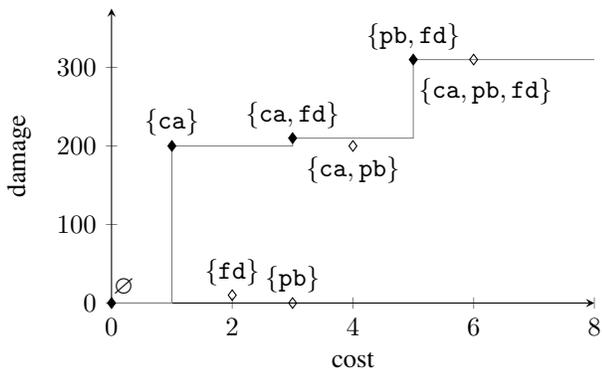
\begin{figure}[t]
\centering
\begin{tikzpicture}
\begin{axis}[width=8cm,height=5.5cm,axis lines = left,xmin = 0,xmax = 8, ymin = 0, ymax = 375,xlabel = {cost},ylabel={damage}]
\addplot[color=black!50, mark = none] table[col sep = comma] {CD_royalty.csv};
\addplot[color=black, mark = diamond*, only marks] table[col sep = comma] {CD_royalty_light.csv};
\addplot[color=black, mark = diamond, only marks] table[col sep = comma] {CD_royalty_no.csv};
\node at (axis cs: 0.2,20) {$\varnothing$};
\node at (axis cs: 2,40) {$\{\mathtt{fd}\}$};
\node at (axis cs: 3,30) {$\{\mathtt{pb}\}$};
\node at (axis cs: 1,230) {$\{\mathtt{ca}\}$};
\node at (axis cs: 3,240) {$\{\mathtt{ca},\mathtt{fd}\}$};
\node at (axis cs: 4,170) {$\{\mathtt{ca},\mathtt{pb}$\}};
\node at (axis cs: 5,340) {$\{\mathtt{pb},\mathtt{fd}\}$};
\node at (axis cs: 6.2,270) {$\{\mathtt{ca},\mathtt{pb},\mathtt{fd}\}$};
\end{axis}
\end{tikzpicture}
\caption{CDPF for Examples \ref{ex:bank2} and \ref{ex:bank3}. Filled nodes are Pareto-optimal attacks.
} \label{fig:PF}
\end{figure}

In what follows, we present novel methods to solve CDPF, DgC and CgD. As in many problems related to calculating AT metrics, an important factor in the complexity of solutions is whether the AT is treelike or not \cite{BS21}. We introduce a bottom-up method for treelike ATs in Section \ref{sec:treedet}, and a method based on integer linear programming for DAG-like ATs in Section \ref{sec:dagdet}.

\section{Relation to knapsack problems and NP-completeness} \label{sec:knapsack}

In this section, we prove two important negative results, based on the similarity of cost-damage problems to binary knapsack problems. First, we show that even the simplest cost-damage problem is NP-complete. Second, we show that cost-damage problems are considerably more general than (extended) knapsack problems, which means that existing heuristics for knapsack problems cannot be applied to our situation. Both results emphasize the importance of finding new heuristics for cost-damage problems.

DgC is a generalisation of the binary knapsack problem \cite{dudzinski_exact_1987}, which is
\begin{align*}
\mathrm{minimize}_{\mathbf{x} \in \BB^n } \ \ f(\mathbf{x}) \ \ \ 
\mathrm{subject\ to} &\ \ g(\mathbf{x}) \leq b
\end{align*}
where $b \in \mathbb{R}$ and $n \in \mathbb{N}$ are constants and the objective and constraint functions $f$ and $g$ are \emph{linear}, i.e., $f(\mathbf{x}) = \sum_{i=1}^n f_i \mathbf{x}_i$ for some constants $f_i \in \mathbb{R}$. In DgC, $n = |B|$, $b = U$, and the objective and constraint functions are $-\hd$ and $\hc$. Although $\hc$ is linear, $-\hd$ is not; for instance, in the AT $\AND(a,b)$, one has $\hd(\mathbf{x}) = \dd(a)x_a+\dd(b)x_b+\dd(\R{T})(x_a \wedge x_b)$. To show NP-completeness, consider the \emph{decision problem} associated to CDPF, DgC and CgD:

\begin{myproblem}[Cost-damage decision problem (CDDP)]
Given a cd-AT $(T,\cc,\dd)$, a cost upper bound $U$ and a damage lower bound $L$, decide whether there exists an attack $\bx \in \FA$ such that $\hc(\bx) \leq U$ and $\hd(\bx) \geq L$. 
\end{myproblem}

CDDP can be reduced to CDPF, DgC or CgD. Theorem \ref{thm:npcomplete} shows that the knapsack decision problem can be reduced to the CDDP (in fact, a \emph{treelike} AT with $n$ BASs and a root suffices). Since the knapsack decision problem is known to be NP-complete \cite{garey1979computers} and it is straightforard to show that CDDP is in NP, we find the following result:

\begin{restatable}{theorem}{thmnpcomplete} \label{thm:npcomplete}
CDDP is NP-complete, even when restricted to treelike ATs.
\end{restatable}

The binary knapsack decision problem (and by extension CDDP) is known to be NP-complete \cite{karp1972reducibility}. It should come as no surprise that we do not give polynomial-time methods to solve CDPF, DgC, and CgD, but instead introduce heuristic methods. These methods discard infeasible solutions throughout the computation instead of at the end, making them faster than the naive approach.

In the literature, many extensions of the binary knapsack problem have been considered that allow less restrictive types of objectives functions, such as quadratic \cite{Gallo1980}, cubic \cite{forrester_strengthening_2022} and submodular \cite{SVIRIDENKO200441} objective functions. However, the following theorem shows that objective functions $\hd$ arising from cd-ATs form the even larger class of \emph{nondecreasing} functions (i.e., $\bx \preceq \by$ implies $f(\bx) \leq f(\by)$, see Definition \ref{def:att}).

\begin{restatable}{theorem}{thmincreasing} \label{thm:increasing}
Let $X$ be a finite set, and let $f\colon \BB^X \rightarrow \mathbb{R}_{\geq 0}$ be any nondecreasing function. Then there is a cd-AT $(T,\cc,\dd)$ with $B = X$ and $\hd = f$.
\end{restatable}

It follows that we cannot use existing binary knapsack approaches to solve DgC, since these approaches \cite{Gallo1980,forrester_strengthening_2022,SVIRIDENKO200441} put some assumptions on $\hd$. 
Instead, we develop new techniques, based on bottom-up methods and integer linear programmming. These techniques exploit the structure of the cd-AT from which the objective $\hd$ originates.

\section{Treelike ATs, deterministic setting} \label{sec:treedet}

For treelike ATs in the deterministic setting we focus on CDPF. DgC and CgD then follow from \eqref{eq:DgC} and \eqref{eq:CgD} respectively. These single-objective problems cannot be computed easier because, as we will demonstrate below, we need to propagate (part of) the Pareto front bottom-up, rather than a single damage/cost value, to solve these problems.

\subsection{CDPF}

A naive way to solve CDPF (and with it DgC and CgD) is by calculating $\hc(\bx)$ and $\hd(\bx)$ for each $\bx \in \FA$. Since $|\FA| = 2^{|B|}$, this is impractical for large ATs, and new heuristics are needed. We solve CDPF via a bottom-up approach in which only a small set of attacks is handled at each node, and infeasilibity is determined at each node rather than at the end. The key insight to make this work is that at intermediate nodes, we perform Pareto analysis in an extended domain $\DB$, and we only project to $\DD$ at the root. 

For a node $v$, we let $T_v$ be the sub-AT of $T$ with root $v$, and we let $B_v$ be its set of BASs. At the node $v$, we are interested in the cost and damage of attacks on $T_v$, which are elements of $\BB^{B_v}$. Suppose that $\ch(v) = \{v_1,v_2\}$. Since $T$ is treelike, one has $B_{v_1} \cap B_{v_2} = \varnothing$. So $\BB^{B_v} = \BB^{B_{v_1}} \times \BB^{B_{v_2}}$, and an attack $\bx$ on $T_v$ can be written $\bx = (\bx_1,\bx_2)$ for attacks $\bx_1$ on $T_{v_1}$ and $\bx_2$ on $T_{v_2}$. With regards to cost and damage, we find
\begin{align}
\hc(\bx) &= \hc(\bx_1) + \hc(\bx_2), \label{eq:dettreecost}\\
\hd(\bx) &= \hd(\bx_1) + \hd(\bx_2)+\struc(\bx,v)\dd(v) \label{eq:dettreedamage}
\end{align}
where we recall that $\struc(\bx,v)$ is defined as
\begin{align}
\struc(\bx,v)&= \begin{cases}
x_v,& \textrm{ if $\type(v) = \BAS$},\\
\struc(\bx_1,v_1) \vee \struc(\bx_2,v_2),& \textrm{ if $\type(v) = \OR$},\\ \struc(\bx_1,v_1) \wedge \struc(\bx_2,v_2),& \textrm{ if $\type(v) = \AND$}.
\end{cases} \nonumber
\end{align}
Thus, in order to correctly calculate the cost and damage of attacks as we combine them, we need to store each attack $\bx$ as an \emph{attribute triple} in the \emph{deterministic attribute triple domain}: \label{page:dtrip}
\begin{equation*} 
\left(\begin{smallmatrix}
\hc(\bx)\\
\hd(\bx)\\
\struc(\bx,v)
\end{smallmatrix}\right) \in \DB := \Rnn \times \Rnn \times \BB.
\end{equation*}

\begin{example} \label{ex:bank2.1}
Consider the AT of Example \ref{ex:bank2}. Each BAS has only two possible attacks (activating that BAS or not) so for $\mathtt{pb}$ we have $\left\{\vvec{0\\0\\0},\vvec{3\\0\\1}\right\} \subset \DB$, and $\left\{\vvec{0\\0\\0},\vvec{2\\10\\1}\right\} \subset \DB$ for $\mathtt{fd}$. Combining these, we have four possible attacks on the $\AND$-gate $\mathtt{dr}$, which is the set
\begin{align*}
\left\{\vvec{0\\0\\0},\vvec{3\\0\\0},\vvec{2\\10\\0},\vvec{5\\110\\1}\right\} \subset \DB.
\end{align*}
\end{example}

After finding the values of all attacks on $v$ by combining those on $v_1$ and on $v_2$, we discard the infeasible ones. Infeasibility is based on two conditions:
\begin{enumerate}
    \item In DgC, if $\hc(\bx) > U$, then $\bx$ is infeasible.
    \item Other than that, feasibility is determined by Pareto optimality on the poset $(\DB,\poD)$, where $\vvec{c\\d\\b} \poD \vvec{c'\\d'\\b'}$ if and only if $c \leq c'$, $d \geq d'$ and $b \geq b'$. The first two inequalities are to be expected from cost-damage optimality. The third inequality is introduced for the following reason: if $\bx$ and $\bx'$ are two attacks on $v$ corresponding to $(c,d,0)^{\intercal}$ and $(c',d',1)^{\intercal}$, respectively, then potentially $\bx'$ can reach nodes higher up in $T$, and thereby eventually do more damage than $\bx$. However, whether this happens or not cannot be detected at the level of $v$, and therefore we need to keep both triples.
\end{enumerate}

\begin{example}\label{ex:bank2.2}
We continue Example \ref{ex:bank2.1}. At $\mathtt{dr}$, we have $\vvec{0\\0\\0} \sqsubset \vvec{3\\0\\0}$, so the latter is infeasible and discarded, leaving us with the Pareto front
\begin{align*}
\left\{\vvec{0\\0\\0},\vvec{2\\10\\0},\vvec{5\\110\\1}\right\} \subset \DB.
\end{align*}
This example shows why we need the third dimension: if not, we would have discarded the attack $\vvec{3\\0}$ at $\mathtt{pb}$ for being infeasible: $\vvec{0\\0}$ does the same damage at lower cost. However, had we done so at $\mathtt{pb}$, we would have concluded that it is always optimal not to activate $\mathtt{pb}$, thereby missing out on the attack $\vvec{5\\110}$ at $\mathtt{dr}$. By also storing the top node's activation, we ensure that activating $\mathtt{pb}$ is still considered feasible. 
\end{example}

This approach can be formally defined as follows. Let $U \in [0,\infty]$. For each $v \in N$, we define a Pareto front $\CU(v) \subseteq \DB$ \label{page:cu} (for \emph{D}eterministic) of feasible attacks on $v$. To do this, we define a map $\mU\colon \Pow(\DB) \rightarrow \Pow(\DB)$ given by \label{page:minu1}
\begin{equation*} 
\mU(X) = \underline{\min}_{\sqsubseteq}  \left\{\vvec{c\\d\\b} \in X: c \leq U\right\}
\end{equation*}
which returns the Pareto-optimal elements (w.r.t. the partial order $\poD$ of $\DB$) of a set $X$ that do not exceed the cost constraint. From now we assume that $T$ is \emph{binary}, i.e., $|\ch(v)| \in \{0,2\}$ for all $v$. Since every AT is equivalent to a binary one this assumption is purely to simplify notation. We then recursively define the Pareto front $\CU(v)$ of attribute triples, by combining elements of $\CU(v_1)$ and $\CU(v_2)$ via \eqref{eq:dettreecost} and \eqref{eq:dettreedamage} and then discarding the nonfeasible triples:
\begin{align*}
&\CU(v) \\
&= \begin{cases}
\left\{\vvec{0\\0\\0},\vvec{\cc(v)\\\dd(v)\\1}\right\}, & \textrm{ if $\type(v) = \BAS$ and } \cc(v) \leq U,\\
\left\{\vvec{0\\0\\0}\right\}, & \textrm{ if $\type(v) = \BAS$ and } \cc(v) > U,
\end{cases} \nonumber\\
&\CU(\AND(v_1,v_2))\\
&= \mU\left\{\vvec{c_1 + c_2 \\d_1+d_2+(b_1 \wedge b_2)\cdot \dd(v) \\b_1 \wedge b_2}\in \DB \middle|\vvec{c_i\\d_i\\b_i\\} \in \CU(v_i)\right\}, \\
&\CU(\OR(v_1,v_2))\\
&= \mU\left\{\vvec{c_1 + c_2 \\d_1+d_2+(b_1 \vee b_2)\cdot \dd(v) \\b_1 \vee b_2} \in \DB\middle|\vvec{c_i\\d_i\\b_i\\} \in \CU(v_i)\right\}. 
\end{align*}
These theorems show the validity of this approach.
\begin{theorem} \label{thm:DgC-tree}
The solution to DgC is given by $\max \left\{d \in \Rnn \middle| \vvec{c\\d\\b} \in \CU(\R{T})\right\}$.
\end{theorem}

\begin{theorem} \label{thm:CDPF-tree}
The solution to CDPF is given by $\min \pi(\mathcal{C}^{\recht{D}}_{\infty}(\recht{R}_T))$, where $\pi\colon \DB\rightarrow \DD$ is the projection map onto the first two components.
\end{theorem}

\begin{example} We continue Examples \ref{ex:bank2.1} and \ref{ex:bank2.2}, for $U = \infty$, in which we calculated $\C{\infty}(\mathtt{dr})$. Below shows the calculation for $\C{\infty}(v)$ for every node; underlined vectors are infeasible and are not part of $\C{\infty}(v)$. The top set is the solution to CDPF.

\vspace{1em}

\begin{tikzpicture}
\node (q0) [rectangle] at (0,0) {$\left\{\vvec{0\\0\\0},\vvec{3\\0\\1}\right\}$} ;
\node at (0,-0.5) {$\mathtt{pb}$};
\node (q1) [rectangle] at (4,0) {$\left\{\vvec{0\\0\\0},\vvec{2\\10\\1}\right\}$} ;
\node at (4,-0.5) {$\mathtt{fd}$};
\node (q3) [rectangle] at (2,1.5) {$\left\{\vvec{0\\0\\0},\sout{\vvec{3\\0\\0}},\vvec{2\\10\\0},\vvec{5\\110\\1}\right\}$};
\node at (2,1) {$\mathtt{rd}$};
\node (q2) [rectangle] at (-2,1.5) {$\left\{\vvec{0\\0\\0},\vvec{1\\0\\1}\right\}$};
\node at (-2,1) {$\mathtt{ca}$};
\node (q4) [rectangle] at (0,3) {$\left\{\vvec{0\\0\\0},\vvec{1\\200\\1},\sout{\vvec{2\\10\\0}},\vvec{3\\210\\1},\vvec{5\\310\\1},\sout{\vvec{6\\310\\1}}\right\}$};
\node at (0,2.5) {$\mathtt{ps}$};
\node (q5) [rectangle] at (0,4.5) {$\left\{\vvec{0\\0},\vvec{1\\200},\vvec{3\\210},\vvec{5\\310}\right\}$};
\draw (q0) edge (q3)
(q1) edge (q3)
(q2) edge (q4)
(q3) edge (q4)
(q4) edge (q5);
\end{tikzpicture}

\vspace{-1em}
\end{example}

\subsection{DgC and CgD}

For DgC we still have to compute a Pareto front at every node $v$, instead of taking the most damaging attack satisfying {the} cost constraint {$\hc(\bx) \leq U$}{, for the following reason. Suppose $\vvec{c\\d\\0},\vvec{c'\\d'\\1}$ are the attribute triples of two attacks on a node $v$ with $c,c' \leq U$ and $d > d'$. If we would just keep the most damaging attack, we would have to discard $\vvec{c'\\d'\\1}$; however, similar to Example \ref{ex:bank2.2}, this could cause us to miss high damage attacks later on in the bottom-up process.}
Thus a `simple' bottom-up approach, in which only a single attack value is propagated, does not work; the best we can do is exclude attacks that at a node already exceed the cost budget. For CgD even this is impossible, as attacks that do not yet satisfy the minimum damage at a certain node may yet do so later.

\subsection{Complexity}

Theorem \ref{thm:complexity} states that the approaches of Theorems \ref{thm:DgC-tree} and \ref{thm:CDPF-tree} are of exponential complexity.

\begin{restatable}{theorem}{thmcomplexity} \label{thm:complexity}
The complexity of solving DgC and CDPF via Theorems \ref{thm:DgC-tree} and \ref{thm:CDPF-tree} is $\mathcal{O}(2^{|B|})$. For CDPF this cannot be improved.
\end{restatable}

The fact that CDPF cannot be computed in less than exponential time can be seen from Example \ref{ex:complexity} below, from the simple reason that the Pareto front may be of exponential size. For DgC, we have improved on the efficiency of CDPF in practice by disregarding attacks that exceed the cost constraint at any node. This does not work for CgD: if an attack does not reach the damage goal at node $v$, that is no reason to regard it as infeasible, because it may be combined with other attacks to increase damage. Therefore we need the full Pareto front to solve CgD in this fashion.

\begin{example} \label{ex:complexity}
Let $T$ be the AT given by $\R{T}= \OR(v_0,\ldots,v_{n-1})$, where $\type(v_i) = \BAS$ and $\cc(v_i) = \dd(v_i) = 2^i$ for all $i < n$; furthermore $\dd(\R{T}) = 0$. Then
\begin{equation*}
\forall \bx \in \mathcal{A}\colon \E(\bx) = \vvec{\sum_{i\colon x_{v_i} = 1} 2^i\\\sum_{i\colon x_{v_i} = 1} 2^i},
\end{equation*}
so each $\E(\bx)$ is optimal in $\E(\FA) = \left\{\vvec{k\\k} \in \DD\middle| k \in \{0,\ldots,2^{n}-1\}\right\}$. It follows that $|\PF| = |\FA| = 2^{|B|}$.
\end{example}

\section{DAG-like ATs, deterministic setting} \label{sec:dagdet}

If $T$ is DAG-like, then the approach outlined in the previous section does not give the correct answer to DgC and CDPF. This is because for a node $v$ with children $v_1$, $v_2$, the sub-ATs $T_{v_1}$ and $T_{v_2}$ may no longer be disjoint, and so equations \eqref{eq:dettreecost} and \eqref{eq:dettreedamage} no longer hold. In particular, if $v_1 = v_2$, we have $\hc(\bx) = \hc(\bx_1) = \hc(\bx_2)$ rather than \eqref{eq:dettreecost}.

\subsection{CDPF}

Instead, to solve CDPF, we introduce a novel approach based on \emph{Biobjective Integer Linear Programming} (BILP) \cite{ozlen2009multi}, i.e., an integer linear programming problem with two objective functions. A BILP problem is of the following form:
\begin{align}
\mathrm{minimise}_{\by \in \mathbb{Z}^n} \ \ \vvec{c_1 \cdot \by \\c_2 \cdot \by} \label{eq:bilp} \ \ \  \mathrm{subject\ to} & \ \ A \cdot \by \leq 0
\end{align}
where $c_1,c_2 \in \mathbb{R}^n$ and $A \in \mathbb{R}^{m \times n}$ for some integers $m,n$. The solution to this BILP problem is the Pareto front 
\begin{equation*}
\underline{\min} \left\{\vvec{c_1 \cdot \by\\c_2 \cdot \by} \in \mathbb{R}^2 \middle| \by \in \mathbb{Z}^n, A\by \leq 0\right\},
\end{equation*}
where $\underline{\min}$ is taken in the poset $(\mathbb{R}^2,\leq)$. Solvers for BILP problems work by repeatedly solving single-objective integer linear programming problems \cite{ozlen2009multi}.
In our case, as variables we use $\by \in \{0,1\}^N$, where we want $y_v$ to represent $\struc(\bx,v)$ for an attack $\bx$ (so $y_v = x_v$ for $v \in B$). Then the objective functions are
\begin{align*}
\hc(\bx) &= \sum_{v \in B} \cc(v) y_v, & -\hd(\bx) &= -\sum_{v \in N} \dd(v) y_v.
\end{align*}
We now need to describe the linear constraints on $\by$. If $v$ is an AND-gate, we introduce a constraint $y_v \leq y_w$ for all children $w$ of $v$; this ensures that $y_v = 0$ whenever at least one of the children has $y_w = 0$. If $v$ is an OR-gate, then we introduce a constraint $y_v \leq \sum_{w \in \recht{ch}(v)} y_w$. Together, these constraints ensure that $y_v \leq \struc(\bx,v)$ for all $v$. Extra constraints that ensure the equality $y_v = \struc(\bx,v)$ are not necessary, because taking $\by$ such that equality holds turns out to be always Pareto optimal. This then leads to the following result:

\begin{restatable}{theorem}{thmbilp} \label{thm:bilp}
CDPF is solved by solving the BILP problem given by
    \begin{align}
    \mathrm{minimise}_{\by \in \{0,1\}^N} \ \ & \left(\begin{array}{c}
    -\sum_{v \in N} \dd(v)y_v \\
    \sum_{v \in B} \cc(v)y_v
    \end{array}\right) \label{eq:detdag}\\
    \mathrm{subject\ to} \ \ & \forall v \in \{v' \in V \mid \type(v') = \AND\}.\nonumber\\
    & \quad \quad \forall w \in \ch(v). \ y_v \leq y_w , \nonumber\\
    & \forall v \in \{v' \in V \mid \type(v') = \OR\}.\nonumber\\
    & \quad \quad \ y_v \leq \sum_{w \in \mathrm{ch}(v)} y_w \nonumber
    \end{align}
\end{restatable}

\begin{example}
Applying Theorem \ref{thm:bilp} to {the AT and cost/damage values of} Example \ref{ex:bank2} yields the following BILP problem:
\begin{align*}
\mathrm{minimise}_{\by \in \{0,1\}^N} \ \ & \vvec{y_{\mathtt{ca}}+3y_{\mathtt{pb}}+2y_{\mathtt{fd}}\\-10y_{\mathtt{fd}-100y_{\mathtt{dr}}-200y_{\mathtt{ps}}}} \\
\mathrm{subject\ to} \ \ & y_{\mathtt{dr}} \leq y_{\mathtt{fb}},\\
& y_{\mathtt{dr}} \leq y_{\mathtt{fd}},\\
& y_{\mathtt{ps}} \leq y_{\mathtt{ca}}+y_{\mathtt{dr}}.
\end{align*}
\end{example}

\subsection{CgD and DgC}

We solve CgD and DgC, by deriving \emph{constrained single-objective optimization problems} from \eqref{eq:detdag}. Associated to a BILP problem {\eqref{eq:bilp}} one {has} these single-objective problems:
\begin{align*}
\mathrm{minimize}_{\by \in \mathbb{Z}^n} & \ \ c_1 \cdot \by & \ \ 
\mathrm{subject\ to} & \ \ A \cdot \by \leq 0, \\ 
&&& \ \ c_2 \cdot \by \leq C_2,\\
\mathrm{minimize}_{\by \in \mathbb{Z}^n} & \ \ c_2 \cdot \by & \ \
\mathrm{subject\ to} & \ \ A \cdot \by \leq 0, \\
&&& \ \ c_1 \cdot \by \leq C_1.
 \end{align*}
These are standard integer linear program (ILP) problems, for which efficient solvers exist \cite{Chen2011}. By applying this to \eqref{eq:detdag}, we can formulate DgC and CgD as single-objective ILP problems, which can be fed to a solver.

\begin{theorem} \label{thm:ilp}
DgC and CgD are solved by solving the constrained single-objective optimization problems derived from \eqref{eq:detdag} with respective added constraints
\begin{align*}
\sum_{v \in B} \cc(v)y_v &\leq U, & -\sum_{v \in N} \dd(v)y_v &\leq -L.
\end{align*}
\end{theorem}

Note that to solve DgC and CgD via Theorem \ref{thm:ilp}, one does not need to first solve the BILP problem \eqref{eq:detdag}, but one can directly solve the single-objective problem.

\section{Probabilistic cost-damage Pareto front} \label{sec:prob}

So far, we have assumed that any BAS undertaken by the attacker will succeed. However, in reality an attempted BAS may or may not succeed. Following earlier work \cite{rauzy1993new,BS21} we now assume a \emph{probabilistic setting} in which each BAS $v$ has a success probability $\pp(v)$. More precisely, we assume:
\begin{enumerate}
    \item The activation of the BASs may or may not succeed;
    \item The successes of different BASs are independent;
    \item The attacker pays the cost of a BAS, whether its activation succeeds or not;
    \item All BASs are attempted simultaneously and paid for in advance; 
    \item Each BAS can only be attempted once.
\end{enumerate}

The independence assumption is standard \cite{rauzy1993new,BS21}, while the other assumptions lead to the most straightforward setting. Extensions are possible: for instance, the attacker might recoup some of the costs of failed activations, or BASs are attempted one by one and the attacker may choose to reallocate their budget based on BASs that have succeeded or failed their activation thusfar. Such extensions lead to more complicated models, and are left to future work. 

\begin{definition} \label{def:cdp}
A \emph{cdp-AT} is a tuple $(T,\cc,\dd,\pp)$ of an AT $T$ and maps $\cc\colon B\rightarrow \Rnn$, $\dd\colon N \rightarrow \Rnn$, and $\pp\colon B \rightarrow [0,1]$.
\end{definition}

In a cdp-AT, the damage done by an attack is a random variable: its value depends on the \emph{actualized attack}, i.e., the BASs that succeed. Therefore, an attacker is interested in the \emph{expected damage} of an attack.

\begin{definition} \label{def:de}
Let $(T,\cc,\dd,\pp)$ be a cdp-AT. For $\bx \in \mathcal{A}$, define the \emph{actualized attack} to be the random variable $Y_{\bx}$ on $\mathcal{A}$ given by
\begin{equation*}
\Prob(\Yx = \by) = \begin{cases}
\prod_{v\colon x_v = 1} \pp(v)^{y_v}(1-\pp(v))^{1-y_v}, & \textrm{ if $\by \preceq \bx$,}\\
0, & \textrm{ otherwise.}
\end{cases}
\end{equation*}
We define the \emph{expected damage} of an attack to be $\hdE(\bx) = \Ex[\hd(\Yx)] \in \Rnn$.
\end{definition}

\begin{example} \label{ex:bank4}
We return to the setting of Example \ref{ex:bank2}. We extend the cd-AT $(T,\cc,\dd)$ with a probability map $\pp\colon B \rightarrow [0,1]$ given by $\pp(\mathtt{ca}) = 0.2$, $\pp(\mathtt{pb}) = 0.4$ and $\pp(\mathtt{fd}) = 0.9$. We use this to calculate the function $\hdE$; we write an attack $\bx$ as the vector $(x_{\mathtt{ca}},x_{\mathtt{pb}},x_{\mathtt{fd}})$. Then the random variable $Y_{(0,1,1)}$ is given by
\begin{align*}
\Prob[Y_{(0,1,1)} = (0,0,0)] &= 0.6 \cdot 0.1 = 0.06,\\
\Prob[Y_{(0,1,1)} = (0,0,1)] &= 0.6 \cdot 0.9 = 0.54,\\
\Prob[Y_{(0,1,1)} = (0,1,0)] &= 0.4 \cdot 0.1 =  0.04,\\
\Prob[Y_{(0,1,1)} = (0,1,1)] &= 0.4 \cdot 0.9 = 0.36.
\end{align*}
\end{example}

Similar to $\hdE$, we also define $\EE(\bx) = (\hc(\bx),\hdE(\bx)) \in \DD$. We then have the following probabilistic counterparts of CDPF, DgC, and CgD:

\hlbox{
\textbf{Problems.}
Given a cdp-AT $(T,\cc,\dd,\pp)$, solve the following problems: \label{prob:cedpf}\label{prob:edgc}\label{prob:cged}
\begin{enumerate}[topsep=0pt,leftmargin=40pt]
    \item[\textbf{CEDPF}] Cost-expected damage Pareto front: find $\umin_{\poD} \EE(\FA) \subseteq \DD$.
    \item[\textbf{EDgC}] Maximal expected damage given cost constraint: Given $U \in \mathbb{R}_{\geq 0}$, find $d_{\recht{E},\recht{opt}} = \max_{\bx\colon \hc(\bx) \leq U} \hdE(\bx)$.
    \item[\textbf{CgED}]  Minimal cost given expected damage constraint: $L \in \mathbb{R}_{\geq 0}$, find $c_{\recht{E},\recht{opt}} = \min_{\bx\colon \hdE(\bx) \geq L} \hc(\bx)$.
\end{enumerate}
}

\begin{example}
We continue Example \ref{ex:bank4}. Using the definition of $\hd$ from the table in Example \ref{ex:bank2}, we find $\hd_{\recht{E}}(0,1,1) = 0.06 \cdot 0 + 0.54 \cdot 0 + 0.04 \cdot 10 + 0.36 \cdot 310 = 112$.
\end{example}

Solving CDEPF naively is more involved than CDPF: not only do we have to calculate $\hdE(\bx)$ for exponentially many $\bx$, but a single $\hdE(\bx)$ also requires $\Prob(\Yx = \by)\cdot \hd(\by)$ for exponentially many $\by$. Therefore, we introduce new methods to solve CDEPF for treelike ATs in Section \ref{sec:treeprob}, by adapting the deterministic method of Section \ref{sec:treedet} to account for probabilities. For DAG-like ATs, we cannot simply adapt the BILP method of Section \ref{sec:dagdet}, as \eqref{eq:detdag} becomes nonlinear, and CDEPF, EDgC and CgED for DAG-like ATs are left to future work.

\section{Treelike ATs, probabilistic setting} \label{sec:treeprob}

EDgC and CEDPF for treelike ATs can be solved similar to the approach of Section \ref{sec:treedet}. The main difference is that instead of working with the structure function $\struc(\bx,v)$, we work with the \emph{probabilistic structure function} $\kk(\bx,v) := \Prob(\struc(Y_{\bx},v) = 1)$.\label{page:ps} With this notation we can write
\begin{equation*}
\hdE(\bx) = \sum_{v \in N} \kk(\bx,v)\dd(v).
\end{equation*}
Let $v$ be a node with children $v_1,v_2$, and let $\bx \in \mathcal{A}$. Since $T$ is treelike, $v_1$ and $v_2$ do not have shared BASs. Since the truth values of the BASs in $Y_{\bx}$ are independent of each other, this means that the random variables $\struc(Y_{\bx},v_1)$ and $\struc(Y_{\bx},v_2)$ are independent, and so we find
\begin{align}
&\kk(\bx,\OR(v_1,v_2)) \nonumber\\
&\ \ = \kk(\bx,v_1)+\kk(\bx,v_2)-\kk(\bx,v_1)\kk(\bx,v_2), \label{eq:kor}\\
&\kk(\bx,\AND(v_1,v_2)) \nonumber\\
&\ \ = \kk(\bx,v_1)\kk(\bx,v_2). \label{eq:kand}
\end{align}
On the other hand, we can express $\hdE(\bx)$ as
\begin{align}
\hdE(\bx) &= \hdE(\bx_1) + \hdE(\bx_2)+\kk(\bx,v)\dd(v). \label{eq:probtreedamage}
\end{align}
Combining this with \eqref{eq:kor} and \eqref{eq:kand} we can calculate the attributes $\hc$, $\hdE$, $\kk$ of attacks on $v$ from their constituent attacks on $v_1$ and $v_2$. From here, we continue akin to Section \ref{sec:treedet}. More precisely, we consider the \emph{probabilistic attribute triple domain}, which is the poset $(\DP,\poD)$ given by \label{eq:ptrip}$\DP = \Rnn \times \Rnn \times [0,1]$ and $(c,d,p) \poD (c',d',p')$ if and only if $c \leq c', d \geq d' \textrm{ and } p \geq p'$. For every node $v$ we define a set $\CpU(v) \subseteq \Pow(\DP)$ of attribute triples. Just as in the deterministic case, we add the requirement $p \geq p'$ in determining feasibility because a greater activation probability of a node may lead to more damage higher up in the AT. As in Section \ref{sec:treedet}, we define a map $\mU\colon \Pow(\DP) \rightarrow \Pow(\DP)$ by\label{page:umin2} $\mU(X) = \umin\left\{\vvec{c\\d\\p} \in X: c \leq U\right\}$. Define $\star\colon [0,1]^2 \rightarrow [0,1]$ by $p \star p' = p+p'-pp'$. Then we again assume that $T$ is binary, and we define $\CpU(v)$ \label{page:cpu} recursively by
\begin{align}
&\CpU(v) \label{eq:probtreebas}\\
&= \begin{cases}
\left\{\vvec{0\\0\\0},\vvec{\cc(v)\\\pp(v)\dd(v)\\\pp(v)}\right\}, & \textrm{ if $\type(v) = \BAS$ and } \cc(v) \leq U,\\
\left\{\vvec{0\\0\\0}\right\}, & \textrm{ if $\type(v) = \BAS$ and } \cc(v) > U,
\end{cases}  \nonumber\\
&\CpU(\OR(v_1,v_2)) \label{eq:probtreeor}\\
&= \mU\left\{\vvec{c_1 + c_2 \\d_1+d_2+(p_1 \star p_2)\cdot \dd(v) \\p_1 \star p_2} \in \DP\middle| \vvec{c_i\\d_i\\p_i\\} \in \CpU(v_i)\right\},  \nonumber\\
&\CpU(\AND(v_1,v_2)) \label{eq:probtreeand}\\
&= \mU\left\{\vvec{c_1 + c_2 \\d_1+d_2+p_1p_2\dd(v) \\p_1p_2}\in \DP \middle| \vvec{c_i\\d_i\\p_i\\} \in \CpU(v_i)\right\}. \nonumber
\end{align}
Then similar to the results in Section \ref{sec:treedet} one can prove:

\begin{restatable}{theorem}{thmEDgCtree} \label{thm:EDgC-tree}
The solution to EDgC is given by $\max \left\{d \in \Rnn \middle| \vvec{c\\d\\p} \in \CpU(\R{T})\right\}$.
\end{restatable}

\begin{restatable}{theorem}{thmCEDPFtree} \label{thm:CEDPF-tree}
The solution to CEDPF is given by $\min \pi(\mathcal{C}_{\infty}^{\recht{P}}(\R{T}))$, where $\pi\colon \DP\rightarrow \DD$ is the projection map onto the first two coefficients.
\end{restatable}

\begin{figure*} 
\includegraphics[width=18cm]{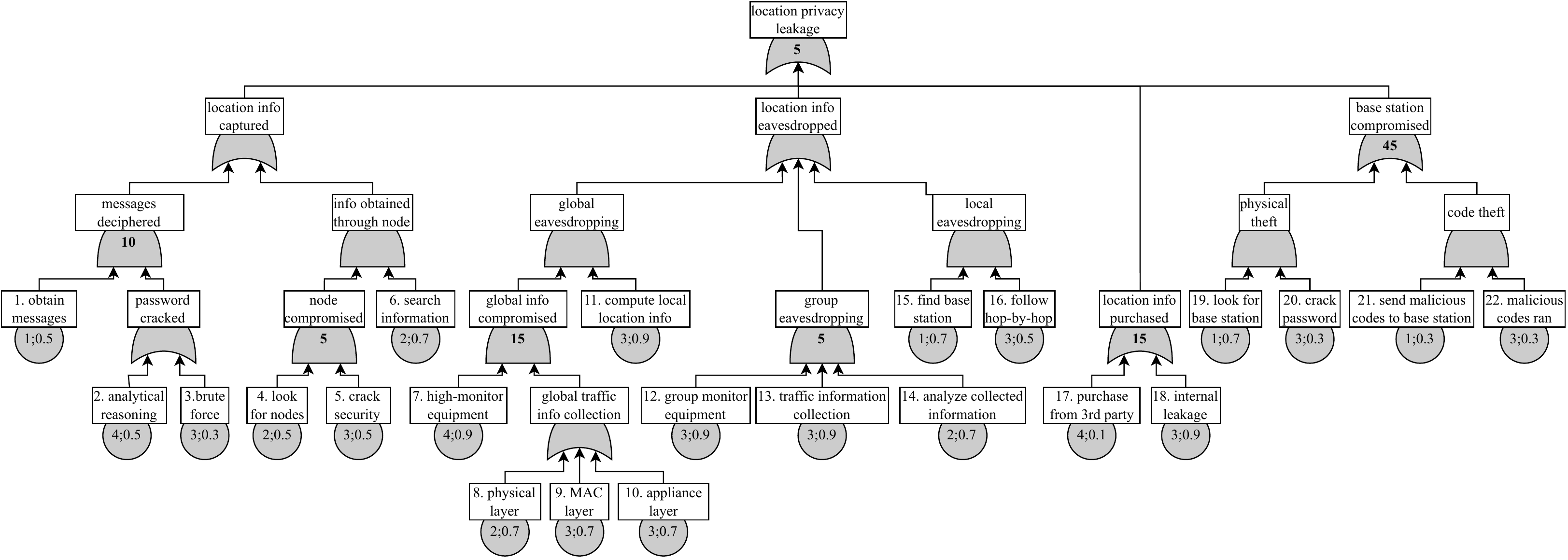}
\caption{Attack tree for privacy attacks on a giant panda preservation IOT monitoring system \cite{jiang2012attack}. Nonzero damage values (in million USD) are in \textbf{bold}, BASs have cost values (unitless) and probabilities inscribed.} \label{fig:panda}
\end{figure*}

In the worst case, the complexity of this approach will be the same as in Section \ref{sec:treedet}; the Pareto frontier can still be of exponential size. Typically, however, $\CpU(v)$ will be larger than $\CU(v)$; in the deterministic model, it is often nonoptimal to add BASs with no damage but with extra costs to an attack, when that attack already activates their parent nodes. However, in the probabilistic model, attempting extra BASs that are not needed in the deterministic model typically leads to a higher probability of activating the parent nodes, giving another way of increasing the cost of an attack to increase its expected damage.

\begin{example} \label{ex:double}
Consider the AT with $w = \R{T} = \OR(v_1,v_2)$, with $\gamma(v_i) = \BAS$, $\cc(v_i) = 1$, $\dd(v_i) = 0$, $\pp(v_i) = 0.5$ for $i = 1,2$, and $\dd(w) = 1$. For $U \geq 2$ the incomplete Pareto fronts $\CU$ and $\CpU$ are given in the table below.

\begin{tabular}{ccc}
Node & $\CU$ & $\CpU$ \\ \hline
$v_1,v_2$ & $\left\{\vvec{0\\0\\0},\vvec{1\\0\\1}\right\}$ & $\left\{\vvec{0\\0\\0},\vvec{1\\0\\0.5}\right\}$\\
$w$ & $\left\{\vvec{0\\0\\0},\vvec{1\\1\\1}\right\}$ & $\left\{\vvec{0\\0\\0},\vvec{1\\0.5\\0.5},\vvec{2\\0.75\\0.75} \right\}$
\end{tabular}
\end{example}

\begin{wrapfigure}[5]{r}{3cm}
\centering
\includegraphics[width=3cm]{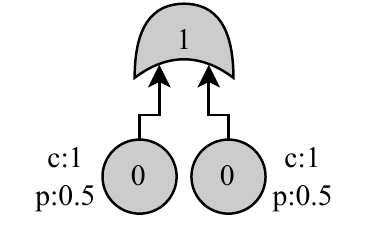}
\end{wrapfigure}
In the deterministic case one $v_i$ suffices to reach $w$, and activating the other comes with extra costs without benefit, which is infeasible. In the probabilistic case attempting both $v_i$ instead comes at the same extra cost, but it increases the expected damage because it increases the probability of $w$ being reached.

For DAG-like ATs in the probabilistic setting one cannot transpose our BILP approach of Section \ref{sec:dagdet}, because the associated equations become nonlinear. For instance, if we introduce a vector $\vec{y} \in [0,1]^N$ where $y_v$ represents $\kk(\mathbf{x},v)$, then for $v = \AND(v_1,v_2)$ we get a constraint $y_v = y_{v_1} \cdot y_{v_2}$, which is nonlinear. In Section \ref{sec:dagdet}, this issue was circumvented because this equation can be linearized if one knows $y_v \in \{0,1\}$, but in general this is not possible. Therefore, we leave CEDPF, CgED and EDgC for DAG-like ATs as an open problem.

\section{Experiments} \label{sec:exp}

\begin{figure}
\includegraphics[width=9cm]{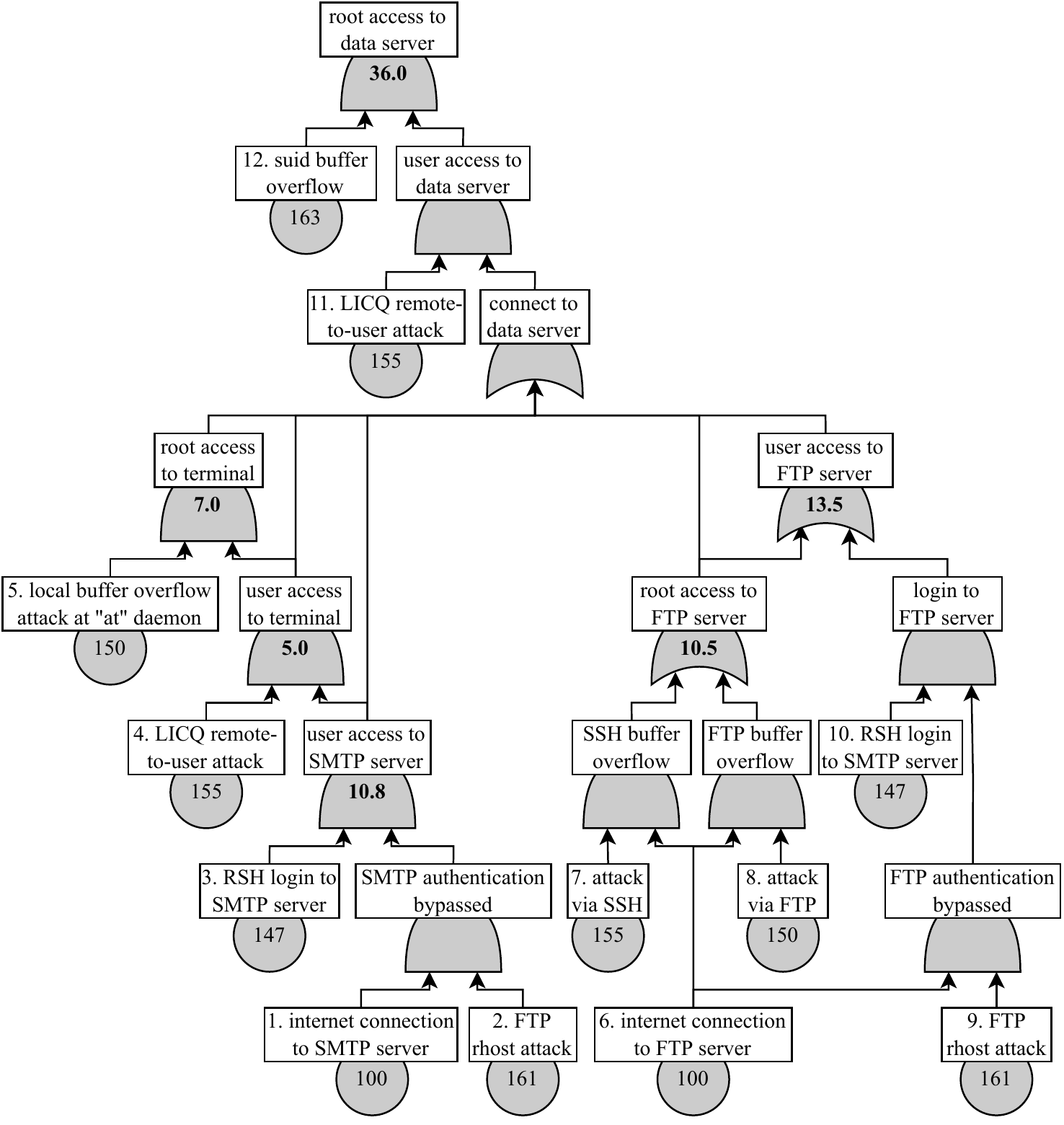}
\caption{Attack tree for a data server on a network \cite{dewri2012optimal}. Nonzero damage values (unitless) are in \textbf{bold}, BASs have cost (in seconds) inscribed.} \label{fig:server}
\end{figure}

\begin{figure}
\begin{subfigure}{0.45\textwidth}
\centering
\begin{tikzpicture}
\begin{axis}[width=7cm,height=4cm,axis lines = left,xmin = 0,xmax = 35, ymin = 0, ymax = 100,ylabel={damage}]
\addplot[color=black!50, mark = none] table[col sep = comma] {CD_panda.csv};
\addplot[color=black, mark = diamond*, only marks] table[col sep = comma] {CD_panda_light.csv};
\end{axis}
\end{tikzpicture}
\begin{tabular}{llccc}
Attack & BASs & cost & damage & {top} \\ \hline 
$A_1$ & $\{b_{18}\}$ & 3 & 20 & {y} \\
$A_2$ & $\{b_{19},b_{20}\}$ or $\{b_{21},b_{22}\}$ & 4 & 50 & {y} \\
$A_3$ & $A_1 \cup A_2$ & 7 & 65 & {y} \\
$A_4$ & $A_3 \cup \{b_1,b_3\}$ & 11 & 75 & {y} \\
$A_5$ & $A_3 \cup \{b_7,b_8\}$ & 13 & 80 & {y} \\
$A_6$ & $A_4 \cup A_5$ & 17 & 90 & {y} \\
$A_7$ & $A_6 \cup \{b_4,b_5\}$ & 22 & 95 & {y} \\
$A_8$ & $A_7 \cup \{b_{11},b_{12},b_{13}\}$ & 30 & 100 & {y}\\
\end{tabular}
\caption{Cost-damage Pareto front for Fig. \ref{fig:panda}.} \label{fig:cdpanda}
\end{subfigure}
\begin{subfigure}{0.45\textwidth}
\centering
\begin{tikzpicture}
\begin{axis}[width=7cm,height=4cm,axis lines = left,xmin = 0,xmax = 60, ymin = 0, ymax = 60,ylabel={expected damage}]
\addplot[color=black!50, mark = none] table[col sep = comma] {CDP_panda.csv};
\addplot[color=black, mark = diamond*, only marks] table[col sep = comma] {CDP_panda_light.csv};
\end{axis}
\end{tikzpicture}
\begin{tabular}{llccc}
Attack & BASs & cost & damage & {top} \\ \hline
$A_1$ & $\{b_{18}\}$ & 3 & 18.0 & {y} \\
$A_2$ & $A_1 \cup \{b_{19},b_{20}\}$ & 7 & 27.6 & {y} \\
$A_3$ & $A_2 \cup \{b_{21},b_{22}\}$ & 11 & 30.8 & {y} \\
$A_4$ & $A_2 \cup \{b_7,b_8\}$ & 13 & 37.0 & {y} \\
$A_5$ & $A_4 \cup \{b_9\}$ & 16 & 39.8 & {y} \\
$\vdots$ & $\vdots$ & $\vdots$ & $\vdots$ \\ 
\end{tabular}
\caption{Cost-expected damage Pareto front for Fig. 
\ref{fig:panda}}
\end{subfigure}
\begin{subfigure}{0.45\textwidth}
\centering
\begin{tikzpicture}
\begin{axis}[width=7cm,height=4cm,axis lines = left,xmin = 0,xmax = 1400, ymin = 0, ymax = 100,ylabel={damage}]
\addplot[color=black!50, mark = none] table[col sep = comma] {CD_server.csv};
\addplot[color=black, mark = diamond*, only marks] table[col sep = comma] {CD_server_light.csv};
\end{axis}
\end{tikzpicture}
\begin{tabular}{llccc}
Attack & BASs & cost & damage & {top} \\ \hline 
$A_1$ & $\{b_6,b_8\}$ & 250 & 24 & {n} \\
$A_2$ & $A_1 \cup \{b_{11},b_{12}\}$ & 568 & 60 & {y} \\
$A_3$ & $A_2 \cup \{b_1,b_2,b_3\}$ & 976 & 70.8 & {y} \\
$A_4$ & $A_3 \cup \{b_4\}$ & 1131 & 75.8 & {y} \\
$A_5$ & $A_4 \cup \{b_5\}$ & 1281 & 82.8 & {y}
\end{tabular}
\caption{Cost-damage Pareto front for Fig. \ref{fig:server}.} \label{fig:cdserver}
\end{subfigure}
\caption{Pareto fronts for the example ATs, together with the corresponding attacks as subsets of $B$. {Except for $A_1$ of (c) all optimal attacks reach the top node.}}
\end{figure}

We tested the validity of our methods by executing them on two established ATs from the literature; these model the attacks on private information of valuable assets in a wireles sensor network \cite{jiang2012attack} and on a data server on a network behind a firewall \cite{dewri2012optimal}. {We also evaluate computation time on a suite of randomly generated ATs. As discussed in Section \ref{sec:rel}, existing approaches cannot be applied to solve the Cg(E)D, (E)DgC and C(E)DPF problems; instead, we compare computation time to an enumerative method that goes through all attacks to find the Pareto optimal ones.}

The methods are implemented in Matlab and executed on PC with an Intel  Core i7-10750HQ 2.8GHz processor and 16GB memory. The source code can be found at \cite{code}. The BILP problems are solved by translating them into single-objective problems via the methods of \cite{ozlen2009multi} in the YALMIP environment \cite{yalmip}, which translates them into the Gurobi solver \cite{gurobi}, a state-of-the-art optimizer that can handle ILP problems.
We find that our methods compute C(E)DPF considerably faster than the naive method, and that the resulting Pareto front provides valuable insight into the weak points of the system.

\subsection{IoT sensor network for wildlife observation}

The first AT \cite{jiang2012attack} is treelike (Fig.~\ref{fig:panda}). It {shows} attacks on a wireless {IoT} sensor network that have the goal of obtaining the location information of valuable assets; in this case, giant pandas in a reservation in China \cite{jiang2012attack}. The costs of BASs are given in \cite{jiang2012attack} as unitless values 1--5. Detection probability is also given as a value 1--5; we take this as the BAS's success probability by converting it to a value 0.1--0.9. The work \cite{jiang2012attack} does not contain damage values; instead, we estimate these from the economic value of giant pandas and the average reservation size \cite{wei2018value}. The top event (the location information of one giant panda) only does minor damage compared to some of the internal nodes; {e.g}, if the base station is compromised, {all pandas'} location information is leaked.

On this AT, we first disregard probability and calculate the cost-damage Pareto front bottom-up via Theorem \ref{thm:CDPF-tree}. The resulting Pareto front is shown in Fig. \ref{fig:cdpanda}, and the corresponding Pareto-optimal attacks are listed as subsets of $B$ (where $b_i$ is the BAS numbered $i$ in Fig. \ref{fig:panda}). As we can see, only a few of the $2^{22}$ possible attacks are Pareto optimal. Furthermore, every optimal attack contains at least one of the minimal attacks $\{b_{18}\}, \{b_{19},b_{20}\}$ and $\{b_{21},b_{22}\}$, and many contain two of them. These three minimal attacks do a lot of damage at relatively small cost; indeed, after these the curve tapers off, and extra cost beyond this has less damage impact. Thus, these attacks require the most defense, {and security improvements should focus on location information leakage by internal sources ($b_{18}$) and base station compromise by either physical theft ($b_{19},b_{20}$) or code theft ($b_{21},b_{22}$). After defenses are put in place, a new cost-damage analysis is needed to see whether attack risks have been mitigated satisfactorily.}

We also calculate the cost-expected damage Pareto front via Theorem \ref{thm:CEDPF-tree}. It has 31 Pareto-optimal attacks; this increase compared to the deterministic situation comes from the fact that in the probabilistic case it is beneficial to activate multiple children of an OR-gate, as in Example \ref{ex:double}. Again the attack $\{b_{18}\}$ is Pareto-optimal at $(3,18)$; however, $\{b_{19},b_{20}\}$ and $\{b_{21},b_{22}\}$ have expected damage $10.5$ and $4.5$, respectively, and at cost $4$ are no longer Pareto-optimal. Instead, the next Pareto-optimal attack is $\{b_{18},b_{19},b_{20}\}$, which targets two valuable low-level nodes. In this probabilistic setting, we see that {internal local information leakage ($b_{18}$)} is part of every Pareto-optimal attack, which suggests this is the most important attack to defend against.

\subsection{Data server on a network}

The second AT we consider represents the attack on a data server through a firewall using known exploits \cite{dewri2012optimal}. Since is DAG-like we only consider the deterministic case. The damage values are from \cite{dewri2012optimal} and represent unitless composites aggregating lost revenue, non-productive downtime, damage recovery, public embarassment, and law penalty. The cost is measured in time spent by the attacker, and the values are taken from \cite{zhao2011hybrid}, where the time taken by similar attacks is modeled via exponential distributions; we take the expected value as each node's duration. The rates in \cite{zhao2011hybrid} are unitless, so we assume they are in $\tfrac{1}{100\recht{s}}$; this does not affect the Pareto front except for scaling. We have slightly changed the AT compared to \cite{dewri2012optimal} as the presentation there focused on vulnerabilities rather than attacks. Note that some nodes, such as $
\mathtt{UserAccessToTerminal}$, are superfluous if one only cares about activating the top node since they require $\mathtt{UserAccessToSMTPServer}$, but they do play a role in cost-damage analysis since they carry damage values.

The results are depicted in Fig. \ref{fig:cdserver}. There are 5 nonzero Pareto-optimal attacks. Furthermore, every Pareto-optimal attack contains the previous one. This implies that {FTP buffer overflow attacks on the FTP server ($b_6,b_8$)} are the most important BASs to defend against, followed by $b_{11}$ and $b_{12}$, etc. Note that of these Pareto optimal attacks only $A_2$ would have been found by a minimal attack analysis.

\subsection{Computation time{: Case studies}}

We also measure the computation time of both our bottom-up and BILP methods for our analyses where applicable. We also compare it to an enumerative approach in which we calculate the cost and damage for each possible attack, and keep only the Pareto-optimal ones. For Fig. \ref{fig:panda}, this amounts to $2^{22} \approx 4\cdot 10^6$ attacks. {The bottom-up method is about $10\times$ as fast as BILP, and both outperform the enumerative method by an enormous margin, especially in the larger Fig.~\ref{fig:panda}.}

To check the robustness of our timing results, we also evaluate our methods on the same ATs, but with random $\cc,\dd,\pp$ values on each node {($\cc(v) \in \{1,\ldots,10\}$, $\dd(v) \in \{0,\ldots,10\}$, $\pp(v) \in \{0.1,0.2,\ldots,1.0\}$)}. The average computation time and standard deviations are given in Table \ref{tab:time}. For the bottom-up methods, the results conform to our earlier results, but \emph{BILP is} slower; this may be because the random ATs contain considerably more nonzero values. {The enumerative method is skipped because it is a lot slower than our new approaches; we compare it to our methods more comprehensively below.}

{
\subsection{Computation time: Randomly generated ATs}

We also apply our methods to a suite of ATs, randomly generated through a method adapted from \cite{lopuhaa2021attack}. More specifically, we generate ATs by taking literature ATs (see Table \ref{tab:DATs}) and combining them in one of the three following ways (see \cite{lopuhaa2021attack}): 
\begin{enumerate}
\item We take a random BAS from the first AT and replace it with the root of the second AT, thus joining the two ATs;
\item We give the roots of the two ATs a common parent with a random type;
\item Same as the previous, but we also identify two randomly chosen BASs, one from each AT. 
\end{enumerate}
For each integer $1 \leq n \leq 100$, we combine ATs from Table \ref{tab:DATs} via a method randomly drawn from the three above until the resulting AT satisfies $|N| \geq n$. We do this five times for each $n$, yielding a testing suite $\mathcal{T}_{\recht{DAG}}$ of 500 DATs, with random $\cc,\dd,\pp$ as above. To test our bottom-up methods,  we also create a suite $\mathcal{T}_{\recht{tree}}$ of treelike ATs, using the first two combining methods above and only the treelike ATs from Table \ref{tab:DATs}.

We evaluate computation times and average the results in groups of ATs grouped by $\lfloor N \rfloor/10$; see Fig.~\ref{fig:lines}. We only evaluated the enumerative method for the first 3 groups. Again, BU is faster than BILP, and both are considerably faster than the enumerative approach. For large ATs probabilistic BU is slower than deterministic BU, which is not yet seen from the case study ($N = 38$). This is probably because not only there are exponentially many attacks to consider, each attack also considers exponentially many actualized attacks to calculate expected damage, see Example \ref{ex:bank4}.

%

\begin{table*}[t]
\centering
\begin{tabular}{l|ccc|ccc}
& \multicolumn{3}{c|}{True $\recht{c},\recht{d}$, ($\recht{p}$)} & \multicolumn{3}{c}{Random $\recht{c},\recht{d}$, ($\recht{p}$)} \\
AT & time (BU) & time (BILP) & time (enumerative) & time (BU) & time (BILP) & time (enumerative) \\ \hline
Fig. \ref{fig:panda} deterministic & 0.044s & 0.438s & {34h} & {0.037s$\pm$0.004s} & {3.144s$\pm$0.526s} & \\
Fig. \ref{fig:panda} probabilistic & 0.047s & n/a & {49h}& {0.051s$\pm$0.012s} & n/a & \\
Fig. \ref{fig:server} deterministic & n/a &  0.380s & 79.53s & n/a &  {1.558s$\pm$0.252s} & {$84.19$s$\pm$4.79s}
\end{tabular}
\caption{{Computation time for} C(E)DPF for the given ATs using bottom-up methods (Theorems \ref{thm:CDPF-tree} \& \ref{thm:CEDPF-tree}), BILP (Theorem \ref{thm:bilp}) and enumerative methods for their given $\cc,\dd,\pp$ values, and average and standard deviations over $100$ random $\cc,\dd,\pp$ values.} \label{tab:time}
\end{table*}

\begin{table}[t]
\centering 
\begin{tabular}{ccc|ccc}
Source & $|N|$ & treelike & Source & $|N|$ & treelike\\
\hline
     \cite{kumar2015quantitative} Fig. 1 & 12 & no & \cite{Arnold2014} Fig. 3 & 8 & yes \\
     \cite{kumar2015quantitative} Fig. 8 & 20 & no & \cite{Arnold2014} Fig. 5 & 21 & yes \\
     \cite{kumar2015quantitative} Fig. 9 & 12 & no & \cite{Arnold2014} Fig. 7 & 25 & yes \\
     \cite{AGKS15} Fig. 1 & 16 & no & \cite{Fraile2016} Fig. 2 & 20 & yes \\
&&& \cite{Kordy2018} Fig. 1 & 15 & yes 
\end{tabular}
\caption{ATs from the literature used as building blocks. The trees from \cite{Fraile2016} and \cite{Kordy2018} are attack-defense trees; only the root component of the attack part was used for these trees.} \label{tab:DATs}
\end{table}

}

\pgfplotstableread{
X Y
0.5000   -0.3982
1.5000   1.7218
2.5000   2.8236
}{\mytablenaive}

\pgfplotstableread{
X Y
0.5 -1.2321
1.5 -1.0257
2.5 -0.9683
3.5 -1.2160
4.5 -1.1846
5.5 -1.1283
6.5 -0.9645
7.5 -0.8010
8.5 -0.6318
9.5 -0.5560
10.5 -0.3959
11.5 -0.2201
}{\mytableBU}

\pgfplotstableread{
X Y
0.5 0.009
1.5  0.3796
2.5  0.5246
3.5  0.6916
4.5  0.8037
5.5  0.9421
6.5  1.0441
7.5  1.1293
8.5 1.2262 
9.5  1.2821
10.5 1.4043
11.5  1.5486
}{\mytableBILP}

\pgfplotstableread{
X Y
0.5000   -0.1379
1.5000   1.7305
2.5000   2.8243
}{\mytablenaiveP}

\pgfplotstableread{
X Y
0.5 -1.7392
1.5 -1.4650
2.5 -1.3802
3.5 -0.9105
4.5 -0.6525
5.5 -0.0813
6.5 0.5714
7.5 1.3261
8.5 1.6097
9.5 2.0278
10.5 2.3851
11.5 2.6109
}{\mytableBUP}

\pgfplotstableread{
X Y
0.5000   -0.1627
1.5000   1.4653
2.5000   2.8597
}{\mytablenaiveDAG}

\pgfplotstableread{
X Y
0.5 -0.0580
1.5  0.3444
2.5  0.5880
3.5  0.8043
4.5  0.8532
5.5  0.9549
6.5  1.0577
7.5  1.1584
8.5 1.2273
9.5  1.2871
10.5 1.3539
11.5  1.4847
}{\mytableBILPDAG}

\begin{figure}[t]
\centering
\begin{subfigure}{0.45\textwidth}
\centering
\begin{tikzpicture}
\begin{axis}[  
width = 8cm,
height = 4cm,
    legend style={legend pos = north east,legend columns = 3},     
    ymin = -2,
    ymax = 3,
    ytick={-2,-1,0,1,2,3},
    yticklabels={$10^{-2}$,$10^{-1}$,$10^0$,$10^1$,$10^2$,$10^3$},
    ylabel near ticks,
    ylabel={mean time (s)}, 
    xmin = 0,
    xmax = 12,
    xtick={0,2,4,6,8,10,12},  
    xticklabels = {8,20,40,60,80,100,121},
    nodes near coords align={vertical},  
    ] 
\addplot
  plot[color=blue,mark=triangle*,mark options = {fill=blue}]  table [x=X,y=Y] {\mytablenaive};
 \addplot
  plot[color=red,mark=square*,mark options = {fill=red}] table [x=X,y=Y] {\mytableBU};
\addplot
  plot[color=brown,mark=*,mark options = {fill=brown}] table [x=X,y=Y] {\mytableBILP};
 \legend{Enum,BU,BILP}  
\end{axis} 
\end{tikzpicture}
\caption{$\mathcal{T}_{\recht{tree}}$ deterministic}
\end{subfigure}
\begin{subfigure}{0.45\textwidth}
\centering
\begin{tikzpicture}
\begin{axis}[  
width = 8cm,
height = 4cm,
    legend style={legend pos = south east,legend columns = 3},     
    ymin = -2,
    ymax = 3,
    ytick={-2,-1,0,1,2,3},
    yticklabels={$10^{-2}$,$10^{-1}$,$10^0$,$10^1$,$10^2$,$10^3$},
    ylabel near ticks,
    ylabel={mean time (s)}, 
    xmin = 0,
    xmax = 12,
    xtick={0,2,4,6,8,10,12},  
    xticklabels = {8,20,40,60,80,100,121},
    nodes near coords align={vertical},  
    ] 
\addplot
  plot[color=blue,mark=triangle*,mark options = {fill=blue}]  table [x=X,y=Y] {\mytablenaiveP};
 \addplot
  plot[color=red,mark=square*,mark options = {fill=red}] table [x=X,y=Y] {\mytableBUP};
 \legend{Enum,BU}  
\end{axis} 
\end{tikzpicture}
\caption{$\mathcal{T}_{\recht{tree}}$ probabilistic}
\end{subfigure}
\begin{subfigure}{0.45\textwidth}
\centering
\begin{tikzpicture}
\begin{axis}[  
width = 8cm,
height = 4cm,
    legend style={legend pos = south east,legend columns = 3},     
    ymin = -2,
    ymax = 3,
    ytick={-2,-1,0,1,2,3},
    yticklabels={$10^{-2}$,$10^{-1}$,$10^0$,$10^1$,$10^2$,$10^3$},
    ylabel near ticks,
    ylabel={mean time (s)}, 
    xmin = 0,
    xmax = 12,
    xtick={0,2,4,6,8,10,12},  
    xticklabels = {8,20,40,60,80,100,115},
    nodes near coords align={vertical},  
    ] 
\addplot
  plot[color=blue,mark=triangle*,mark options = {fill=blue}]  table [x=X,y=Y] {\mytablenaiveDAG};
 \addplot
  plot[color=brown,mark=*,mark options = {fill=brown}] table [x=X,y=Y] {\mytableBILPDAG};
 \legend{Enum,BILP}  
\end{axis} 
\end{tikzpicture}
\caption{$\mathcal{T}_{\recht{DAG}}$ deterministic}
\end{subfigure}
\begin{subfigure}{0.45\textwidth}
\centering
\begin{tabular}{ll|ccc}
 & & BU & BILP & enum${}^*$ \\ \hline
\multirow{3}{*}{$\mathcal{T}_{\recht{tree}}$ det.} & min & $<$0.01s & 0.563s & 0.266s \\
& mean & 0.159s & 11.12s & 86.87s \\
& max & 1.141s & 44.86s & 3917s \\ \hline
\multirow{3}{*}{$\mathcal{T}_{\recht{tree}}$ prob.} & min & $<$0.01s &  & 0.313s \\
& mean & 42.30s &  & 426.0s\\
& max & 1335s &  & 3853s\\ \hline
\multirow{3}{*}{$\mathcal{T}_{\recht{DAG}}$ prob.} & min & & 0.781s & 0.313s \\
& mean & & 11.09s & 296.3s \\
& max & &  50.08s & 5619s \\
\end{tabular}
\caption{Overall statistics. ${}^*$Only ATs with $N < 30$.}
\end{subfigure}
\caption{Computation time on randomly generated ATs. Means over subsets grouped by $\lfloor N/10 \rfloor$.}\label{fig:lines}
\end{figure}

\section{Conclusion}

This paper introduced two novel methods to solve cost-damage problems for attack trees, both by optimizing damage (resp. cost) under a cost (resp. damage) constraint, and by calculating the cost-damage Pareto front. For treelike ATs, this is done via bottom-up methods, both in the deterministic and the probabilistic case. For DAG-like ATs in the deterministic case, we introduce a method based on integer linear programming.

There are multiple avenues for further research. An obvious one is the probabilistic case for DAG-like ATs, which is not discussed in this paper. One approach would be to use a bottom-up approach, but in a polynomial ring with formal variables for nodes that occur multiple times, rather than in the real numbers. In that way, one can keep track of which nodes occur twice, and tweak addition to prevent double counting. Another extension is to compare the formal, provably optimal approach presented in this paper with a genetic algorithm approach to multiobjective optimization to approximate the Pareto front \cite{ali2020quality}. From experiments it could be established to what extent the performance gain (if any) from using genetic algorithms comes at an accuracy cost. Finally, the cost and damage values may not be precisely known, but carry some uncertainty. A more elaborate analysis can incorporate this uncertainty, for example using fuzzy numbers, to obtain a robust version of the cost-damage Pareto front.

\printbibliography

\appendix

\subsection{Proof of Theorem \ref{thm:increasing}} \label{app:increasing}

\thmincreasing*

\begin{proof}
 Define $n = |X|$. Let $f\colon \mathbb{B}^X \rightarrow \mathbb{R}_{\geq 0}$ be a nondecreasing function, and let $\bx^1,\ldots,\bx^{2^{n}}$ be an ordering of $\mathbb{B}^n$ such that $f(\bx^i) \leq f(\bx^{i+1})$ for all $i < 2^n$, and such that $\bx^i \preceq \bx^j$ implies $i \leq j$ for all $i$ and $j$. Since $f$ is nondecreasing, these two conditions can be fulfilled simultaneously. Furthermore, define an AT $T$ by having $B = X$ and non-leaf nodes $\{A_i,O_i\}_{i\leq 2^{n}}$ and $\R{T}$ given by
 \begin{align*}
 A_i &= \AND(\{v \in B | x^i_v = 1\}),\\
 O_j &= \OR(\{A_i | i \geq j\}),\\
 \R{T} &= \AND(\{O_j \mid j \leq 2^{n}\}).
 \end{align*}
 Furthermore, we define $d\colon N \rightarrow \mathbb{R}_{\geq 0}$ by
 \begin{align*}
 \dd(v) &= 0 & \forall v \in B,\\
 \dd(A_i) &= 0 & \forall i \leq 2^n,\\
 \dd(O_1) &= f(\bx^1),\\
 \dd(O_{j+1}) &= f(\bx^{j+1})-f(\bx^j) & \forall j < 2^n,\\
 \dd(\recht{R}_T) &= 0.
 \end{align*}
 Note that $\recht{R}_T$ does not play a role in the cost-damage analysis and is here purely to satisfy the condition of $T$ having a root. Now consider an attack $\bx_i$; then $\struc(\bx^i,A_j) = 1$ if and only if $\bx^j \preceq \bx^i$. It follows that $\struc(\bx^i,O_j) = 1$ if and only if there is an $\bx^k$ with $j \leq k$ and $\bx^j \preceq \bx^i$. The latter can only happen when $j \leq i$, and so $\struc(\bx^i,O_j) = 1$ if and only if $j \leq i$. It follows that
 \begin{equation*}
 \hd(\bx^i) = \sum_{j \leq i} \dd(O_j) = f(\bx^i)
 \end{equation*}
 and since this holds for all $i$, we have $f = \hd$.
 \end{proof}

\subsection{Proof of Theorem \ref{thm:npcomplete}}

\thmnpcomplete*

\begin{proof}
{First, we show that CDDP is in NP. A witness of a CDDP problem determined by $(T,\cc,\dd,U,L)$ is given by an attack $\bx \in A$; to verify this we need to calculate $\hc(\bx)$ and $\hd(\bx)$. By Definition \ref{def:cd}, the former can be calculated in $\mathcal{O}(|B|)$ time. The latter can be calculated in $\mathcal{O}(|N|+|E|)$ time, as calculating $\struc{}(\bx,v)$ for all $v$ takes $\mathcal{O}(|N|+|E|)$ time via Definition \ref{def:sf}. We conclude that CDDP is in NP.}

Consider a binary knapsack decision problem, i.e. two linear functions $f,g\colon \BB^n \rightarrow \Rnn$, an upper bound $U$ and a lower bound $L$; the problem is to determine whether there exists a $\bx \in \BB^n$ such that $f(\bx) \geq L$ and $g(\bx) \leq U$. This problem is known to be NP-complete \cite{garey1979computers}, and we prove that the cost-damage decision problem by transforming the binary knapsack decision problem into it.

Since $f$ and $g$ are linear and their images lie in $\mathbb{R}_{\geq 0}$, there exist $f_1,\ldots,f_n,g_1,\ldots,g_n \in \Rnn$ such that $f(\bx) = \sum_{i=1}^n f_ix_i$ and $g(\bx) = \sum_{i=1}^n f_ix_i$. Now define a cd-AT $(T,\cc,\dd)$ by taking $N = \{v_1,\ldots,v_n,\R{T}\}$, with the $v_i$ BASs and $\R{T} = \AND(v_1,\ldots,v_n)$. Furthermore, take $\cc(v_i) = g_i$, and $\dd(v_i) = f_i$ and $\dd(\R{T}) = 0$. Then $\hc(\bx) = g(\bx)$ and $\hd(\bx) = f(\bx)$, and so a solution to the cost-damage decision problem is a solution to the binary knapsack decision problem. Since $(T,\cc,\dd)$ is polynomial (even linear) in the size of the original knapsack problem, this shows that the cost-damage decision problem is NP-complete.
 \end{proof}

\subsection{Proof of Theorem \ref{thm:complexity}} \label{app:complexity}

Before we prove the theorem, we first consider the following auxiliary lemma.

\begin{lemma} \label{lem:compaux}
 Let $T = (N,E)$ be a binary tree, i.e., $\ch(v) \in \{0,2\}$ for all $v \in N$. For a non-leaf node $v$, let $b(v)$ be the number of leaf descendants of $v$ (with edges pointing away from the root). Let $v_1,\ldots,v_K$ be an enumeration of the non-leaf nodes of $T$ such that $b(v_i) \leq b(v_j)$ whenever $i \leq j$. Then $b(v_i) \leq i+1$.
 \end{lemma}

\begin{proof}
 We prove this by induction on $|N|$; it is clearly true if $|N| = 3$, where the single internal node, the root, has two leaves as children. Let $T$ be a full binary tree with $|N| > 3$, and let $\R{T}$ be its root; let $a_1,a_2$ be its children. Furthermore, let $T_i = (N_i,E_i)$ be the subtree consisting of $a_i$ and its descendants. Let $v_1,\ldots,v_K$ be an enumeration of the non-leaf nodes of $N$ such that $b(v_i) \leq b(v_j)$ whenever $i \leq j$. Let $i \leq K$; we aim to prove that $b(v_i) \leq i+1$. If $v_i \in N_1$, define $r := |\{j \leq i \mid v_j \in N_1\}|$. Then $r \leq i$. On the other hand, if we restrict the sequence $v_1,\ldots,v_K$ to just elements of $N_1$, the resulting sequence satisfies the conditions of the Lemma for $T_1$, and $v_i$ is the $r$-th element in this sequence. It follows from the induction hypothesis that $b(v_i) \leq r+1$; hence certainly $b(v_i) \leq i+1$. The case that $v_i \in N_2$ is handled similarly. The last remaining case is $v_i = \R{T}$; but this happens when $i=K$ as $\R{T}$ is necessarily the last element of the sequence. Since $T$ is a binary tree, one has $K = |B|-1$, from which this case also follows.
 \end{proof}

\thmcomplexity*

\begin{proof}
 At a node $v$ with children $v_1$, $v_2$, we have to do at most $|\CU(v_1)|\cdot|\CU(v_2)|$ computations. We also have $|\CU(v)| \leq |\CU(v_1)|\cdot|\CU(v_2)|$. Since $|\CU(v)| \leq 2$ if $v \in B$, one straightforwardly proves by induction that for inner nodes one has $|\CU(v)| \leq 2^{b(v)}$, where $b(v)$ is the number of BAS descendants of $v$ as in Lemma \ref{lem:compaux}. It also follows that at $v$ we have to do at most $2^{b(v)}$ computations. Let $v_1,\ldots,v_{|B|-1}$ be an enumeration of $N \setminus B$ such that $b(v_i) \leq b(v_j)$ whenever $v_i \leq v_j$; then the total number of computations is equal to
 \begin{equation}
 \sum_{i=1}^{|B|-1} 2^{b(v)} \leq \sum_{i=1}^{|B|-1} 2^{i+1} = 2^{|B|+1}-2,
 \end{equation}
 where the inequality follows from Lemma \ref{lem:compaux}. This shows the statement about complexity. The fact that this cannot be improved for CDPF follows from Example \ref{ex:complexity}, which shows that the Pareto front can be of size $2^{|B|}$; in particular, outputting it takes at least that much time.
 \end{proof}

\subsection{Proofs of Theorems \ref{thm:bilp} and \ref{thm:ilp}} \label{app:bilp}

We only prove Theorem \ref{thm:bilp}, as the proof of Theorem \ref{thm:ilp} is essentially the same but less involved, as the optimization is only done in one dimension.

\thmbilp*

\begin{proof}
Let $\by \in \BB^N$ satisfy the conditions of \eqref{eq:bilp}, and let $\bx \in \BB^B$ be the attack given by $x_v = y_v$ for $v \in B$. We claim that $y_v \leq \struc(\bx,v)$ for all $v \in N$, and we prove this by induction; clearly it is true for BASs. Suppose the claim is true for $v_1,\ldots,v_n$, and consider $v = \AND(v_1,\ldots,v_n)$. Then \eqref{eq:bilp} amounts to $y_v \leq \min\{y_{v_i} \mid i \leq n\}$. By the induction hypothesis we then have
\begin{equation}
y_v \leq \min\{y_{v_i} \mid i \leq n\} \leq \min\{\struc(\bx,v_i) \mid i \leq n\} = \struc(\bx,v).
\end{equation}
Similarly, if $v = \OR(v_1,\ldots,v_n)$, we get
\begin{align*}
y_v &\leq \min\left\{1,\sum_{i\leq n} y_{v_i}\right\}\\
&\leq \min\left\{1,\sum_{i \leq n} \struc(\bx,v_i)\right\}\\
&= \begin{cases}
0, & \textrm{ if $\struc(\bx,v_i)=0$ for all $i$},\\
1, & \textrm{ otherwise}
\end{cases}\\
&= \struc(\bx,v).
\end{align*}
This proves the claim.

We now continue with the proof of the theorem. Let $\mathcal{F}$ be the set of $\by$ satisfying the conditions of \eqref{eq:bilp}, and write for $\by \in \mathcal{F}$:
\begin{align*}
f(\by) = \sum_{v \in B} \cc(v)y_v, \quad \quad g(\by) = -\sum_{v \in N} \dd(v)y_v.
\end{align*}
Our aim is to prove the following equality:
\begin{equation} \label{eq:bilppf}
\left\{\vvec{c\\-d} \middle| \vvec{c\\d} \in \PF\right\} = \umin \left\{\vvec{f(\by)\\g(\by)} \middle| \by \in \mathcal{F}\right\},
\end{equation}
where the $\umin$ on the RHS is taken in the poset $(\mathbb{R}^2,\leq)$. We first prove ``$\supseteq$''. Let $\by \in \mathcal{F}$ be such that $\vvec{f(\by)\\g(\by)}$ is minimal. Let $\bx \in \FA$ be such that $x_v = y_v$ for all $v \in B$, and let $\by' \in \BB^N$ be given by $y'_v = \struc(\bx,v)$ for all $v$. A straightforward induction proof shows that $\by' \in \mathcal{F}$, and by the claim we have $y_v \leq y'_v$ for all $v \in N$, with equality when $v \in B$. It follows that $f(\by) = f(\by')$ and $g(\by) \geq g(\by')$. Since $\vvec{f(\by)\\g(\by)}$ is minimal, this must be an equality, and we get
 \begin{equation*}
 \vvec{f(\by)\\g(\by)} = \vvec{f(\by')\\g(\by')} = \vvec{\hc(\bx)\\-\hd(\bx)}.
 \end{equation*}
 It remains to be shown that $\vvec{\hc(\bx)\\\hd(\bx)} \in \PF$. Let $\bx'' \in \FA$ be such that $\hc(\bx') \leq \hc(\bx)$ and $\hd(\bx'') \geq \hd(\bx)$. Let $\by'' \in \BB^N$ be such that $y''_v = \struc(\bx'',v)$ for all $v$; then again $\by'' \in \mathcal{F}$, and
 \begin{equation*}
 \vvec{f(\by'')\\g(\by'')} = \vvec{\hc(\bx'')\\-\hd(\bx'')} \leq \vvec{\hc(\bx)\\-\hd(\bx)} = \vvec{f(\by)\\g(\by)}.
 \end{equation*}
 Since $\vvec{f(\by)\\g(\by)}$ is minimal, this means that equality must hold here; this shows that $\bx$ is Pareto optimal, and so we have shown ``$\supseteq$'' in \eqref{eq:bilppf}.

The argument to prove ``$\subseteq$'' is very similar. Let $\bx \in \FA$ be such that $\vvec{\hc(\bx)\\ \hd(\bx)} \in \PF$, and let $\by \in \BB^N$ be given by $y_v = \struc(\bx,v)$. Then $\by \in \mathcal{F}$ and $\vvec{f(\by)\\g(\by)} = \vvec{\hc(\bx)\\-\hd(\bx)}$; we need to show that this vector is minimal. Let $\by' \in \mathcal{F}$ be such that $\vvec{f(\by')\\g(\by')} \leq \vvec{f(\by)\\g(\by)}$. Define $\bx'' \in \FA$ by $x''_v = y'_v$ for all $v \in B$, and define $\by'' \in \BB^N$ by $\by''_v = \struc(\bx'',v)$ for all $v \in N$. Similar to the above we have $\by'' \in \mathcal{F}$ and $\vvec{f(\by'')\\g(\by'')} \leq \vvec{f(\by')\\g(\by')}$. It follows that we have
 \begin{equation*}
 \vvec{\hc(\bx'')\\-\hd(\bx'')} = \vvec{f(\by'')\\g(\by'')} \leq \vvec{f(\by')\\g(\by')} \leq \vvec{f(\by)\\g(\by)} = \vvec{\hc(\bx)\\-\hd(\bx)}.
 \end{equation*}
 Since $\bx$ is Pareto optimal by assumption, the above must have equalities throughout. In particular $\vvec{f(\by')\\g(\by')} = \vvec{f(\by)\\g(\by)}$, which proves that $\by$ is minimal. This shows ``$\subseteq$'' in \eqref{eq:bilppf}, completing the proof.
 \end{proof}

\subsection{Proofs of Theorems \ref{thm:DgC-tree}, \ref{thm:CDPF-tree}, \ref{thm:EDgC-tree} and \ref{thm:CEDPF-tree}} \label{app:bu}

We will show that all these theorems follow from a shared main result, namely Theorem \ref{thm:aux} below. In order to formulate it, we first need a little more notation. For a node $v$, we let $T_v = (N_v,E_v)$ be the sub-DAG of $T$ consisting of $v$ and all its descendants, together with its set of BASs $B_v$, set of attacks $\mathcal{A}_v$, cost, damage, and expected damage functions $\hc_v$, $\hd_v$ and $\hd_{\textrm{E},v}$:
 \begin{align*}
 N_v &= \{w \in N \mid \exists \textrm{ path } v \rightarrow w\},\\
 E_v &= E \cap (N_v \times N_v),\\
 T_v &= (N_v,E_v),\\
 B_v &= B \cap N_v,\\
 \mathcal{A}_v &= \mathbb{B}^{B_v},\\
 \hc_v(\mathbf{x}) &= \sum_{w \in B_v} \cc(w)x_w, & \textrm{ for $\mathbf{x} \in \mathcal{A}_v$},\\
 \hd_v(\mathbf{x}) &= \sum_{w \in N_v} \dd(w)\struc(w),& \textrm{ for $\mathbf{x} \in \mathcal{A}_v$},\\
 \hd_{\textrm{E},v}(\mathbf{x}) &= \mathbb{E}\left[\hd_v(Y_{\mathbf{x}})\right], & \textrm{ for $\mathbf{x} \in \mathcal{A}_v$}.
 \end{align*}
 Furthermore, we let $\kk_v\colon \mathcal{A}_v \times N_v \rightarrow \BB$ be the probabilistic structure function of $T_v$, analogous to Definition \ref{def:sf}. Based on the equations above we define, for a node $v$, the \emph{extended expected attribute map} $\recht{EA}_v\colon \mathcal{A}_v \rightarrow \DP$ by
 \begin{equation*}
 \recht{EA}(\bx) = \vvec{\hc_v(\bx)\\ \hd_{\recht{E},v}(\bx) \\ \kk(\bx,v)}.
 \end{equation*}

Then as we will show below, Theorems \ref{thm:DgC-tree}, \ref{thm:CDPF-tree}, \ref{thm:EDgC-tree} and \ref{thm:CEDPF-tree} all follow from the following result:

\begin{theorem} \label{thm:aux}
 For every $v \in N$ one has
 \begin{equation*}
 \CpU(v) = \underline{\min} \left\{\recht{EA}_v(\mathbf{x}) \in \DP \middle| \mathbf{x} \in \mathcal{A}_v, \hc_v(\mathbf{x}) \leq U\right\}.    
 \end{equation*}
 \end{theorem}

To prove this theorem we first need a few auxiliary lemmas, as well as some definitions. For $\vvec{c_1\\d_1\\p_1},\vvec{c_i\\d_i\\p_i} \in \DP$ and $d \in \Rnn$, we define

\begin{align*}
 \vvec{c_1\\d_1\\p_1} \triangle_d \vvec{c_2\\d_2\\p_2} &= \vvec{c_1+c_2\\d_1+d_2+p_1p_2d\\p_1p_2},\\
 \vvec{c_1\\d_1\\p_1} \triangledown_d \vvec{c_2\\d_2\\p_2} &= \vvec{c_1+c_2\\d_1+d_2+(p_1 \star p_2)d \\ p_1 \star p_2}.
 \end{align*}

Slightly abusing notation, we write $X \triangle_d Y = \left\{\vvec{c_1\\d_1\\p_1} \triangle_d \vvec{c_2\\d_2\\p_2} \middle| \vvec{c_1\\d_1\\p_1} \in X, \vvec{c_2\\d_2\\p_v} \in Y\right\}$ and likewise for $\triangledown_d$. We then have the following result:

\begin{lemma} \label{lem:aux1}
 For an internal node $v \in N\setminus B$ one has
 \begin{align}
 &\recht{EA}_v(\mathcal{A}_v)\nonumber \\
 &= \begin{cases}
 \recht{EA}_{v_1}(\mathcal{A}_{v_1}) \triangle_{\dd(v)} \recht{EA}_{v_2}(\mathcal{A}_{v_2}), & \textrm{ if $\gamma(v) = \AND$},\\
 \recht{EA}_{v_1}(\mathcal{A}_{v_1}) \triangledown_{\dd(v)} \recht{EA}_{v_2}(\mathcal{A}_{v_2}), & \textrm{ if $\gamma(v) = \OR$}.
 \end{cases} \label{eq:aux1}
 \end{align}
 \end{lemma}

\begin{proof}
 Suppose $\gamma(v) = \AND$. Since $T$ is treelike, one has $B_v = B_{v_1} \cup B_{v_2}$ and $B_{v_1} \cap B_{v_2} = \varnothing$. It follows that we can identify $\FA_v = \BB^{B_v} = \BB^{B_{v_1}} \times \BB^{B_{v_2}} = \FA_{v_1} \times \FA_{v_2}$. Let $\bx_1 \in \FA_{v_1}$ and $\bx_2 \in \FA_{v_2}$, and let $\bx = (\bx_1,\bx_2)$ be the corresponding element of $\FA$, i.e., $x_w = x_{i,w}$ when $w \in B_{v_i}$ for $i=1,2$. Then the attributes of $\bx$ are given by
 \begin{align*}
 \hc_{v}(\bx) &= \sum_{w \in B_v} x_w\cc(w)\\
 &= \sum_{w \in B_{v_1}} x_{1,w}\cc(w) + \sum_{w \in B_{v_2}} x_{2,w}\cc(w)\\
 &= \hc_{v_1}(\bx_1) + \hc_{v_2}(\bx_2),\\
 \kk_v(\bx,v) &= \kk_v(\bx,v_1)\kk_v(\bx,v_2)\\
 &= \kk_{v_1}(\bx_1,v_1)\kk_{v_2}(\bx_2,v_2),\\
 \hd_v(\bx) &= \sum_{w \in B_v} \kk_v(\bx,w)\dd(w)\\
 &= \sum_{w \in B_{v_1}} \kk_{v_1}(\bx_1,w)\dd(w)\\
 & \quad \quad +\sum_{w \in B_{v_1}} \kk_{v_2}(\bx_2,w)\dd(w)+\kk_v(\bx,v)\dd(v)\\
 &= \hd_{v_1}(\bx_1)+\hd_{v_2}(\bx_2)+\kk_{v_1}(\bx_1,v_1)\kk_{v_2}(\bx_2,v_2)\dd(v).
 \end{align*}
 We can write this more succinctly as
 \begin{equation*}
 \recht{EA}_v(\bx) = \recht{EA}_{v_1}(\bx_1) \triangle_{\dd(v)} \recht{EA}_{v_2}(\bx_2).
 \end{equation*}
 Ranging over all $\bx_1$ and $\bx_2$ (and consequently over all $\bx$) now proves the lemma. The case that $\gamma(v) = \OR$ is completely analogous.
 \end{proof}

 Furthermore, for $X \subseteq \DP$, we define 
 \begin{equation*}
 H_U(X) := \left\{\vvec{c\\d\\p} \in X \middle| c \leq U\right\}.
 \end{equation*}
 This function, along with the known function $\umin$, satisfies the following properties:

\begin{lemma} \label{lem:aux2}
 For $X,Y \subseteq \DP$ and $d \in \Rnn$ the following hold:
 \begin{align}
 H_U(\umin(X)) &= \umin(H_U(X)), \label{eq:auxeq1}\\
 H_U(X \triangle_d H_U(Y)) &= H_U(X \triangle_d Y), \label{eq:auxeq2}\\
 H_U(X \triangledown_d H_U(Y)) &= H_U(X \triangledown_d Y), \label{eq:auxeq3}\\
 \umin(X \triangle_d \umin(Y)) &= \umin(X \triangle_d Y), \label{eq:auxeq4}\\
 \umin(X \triangledown_d \umin(Y)) &= \umin(X \triangledown_d Y). \label{eq:auxeq5}
 \end{align}
 \end{lemma}

\begin{proof}
 We tackle these equations one by one, starting with \eqref{eq:auxeq1}. Let $x = \vvec{c\\d\\p} \in H_U(\umin(X))$; then $x \in \umin(X)$. Furthermore, $c \leq U$, and so $x \in H_U(X)$. Suppose that $x \notin \umin(H_U(X))$; then there exists an $x' \in H_U(X)$ such that $x' \sqsubset x$. But this contradicts the fact that $x \in \umin(X)$, and such $x'$ cannot exist; this proves $H_U(\umin(X)) \subseteq \umin(H_U(X))$.

Now let $x = \vvec{c\\d\\p} \in \umin(H_U(X))$. Suppose $x \notin \umin X$; then there exists an $x' = \vvec{c'\\d'\\p'}$ such that $x' \sqsubset x$. By definition of $\sqsubset$ this means that $c' \leq c$, and so $x' \in H_U(X)$; but then this contradicts the fact that $x \in \umin(H_U(X))$. We can conclude that $x \in \umin(X)$. Since $x \in H_U(X)$, we also know that $c \leq U$, and so $x \in H_U(\umin(X))$. This proves $H_U(\umin(X)) \supseteq \umin(H_U(X))$.

Next, we prove \eqref{eq:auxeq2} (equation \eqref{eq:auxeq3} is proven analogously to this and is therefore skipped). Since $H_U(Y) \subseteq Y$ it is clear that $H_U(X \triangle_d H_U(Y)) \subseteq H_U(X \triangle_d Y)$; we now prove the opposite direction. Let $x = \vvec{c_1\\d_1\\p_1} \in X$, $y = \vvec{c_2\\d_2\\p_2} \in Y$ be such that $x \triangle_d y \in H_U(X \triangle_d Y)$. Then $c_1+c_2 \leq U$, and since $c_1,c_2 \in \Rnn$ this implies that $c_2 \in U$. Hence $y \in H_U(Y)$, and this proves $H_U(X \triangle_d H_U(Y)) \supseteq H_U(X \triangle_d Y)$.

Finally, we consider \eqref{eq:auxeq4} and \eqref{eq:auxeq5}; since they can be proven completely analogous we only consider the former. First, we note that $\triangle_d$ preserves $\sqsubseteq$ in the following sense: if $y \sqsubseteq y'$, then $x \triangle_d y \sqsubseteq x \triangle_d y'$ for all $x,y,y' \in \DP$. Now let $x \in X$ and $y \in \umin(Y)$ be such that $x \triangle_d y \in \umin(X \triangle_d \umin(Y))$. Suppose that $x \triangle_d y \notin \min(X \triangle_d Y)$; then there exist $x' \in X$, $y' \in Y$ such that $x' \triangle_d y' \sqsubset x \triangle_d y$. Let $y'' \in \umin(Y)$ be such that $y'' \sqsubseteq y'$; then
 \begin{equation*}
 X \triangle_d \umin(Y) \ni x' \triangle_d y'' \sqsubseteq x' \triangle_d y' \sqsubset x \triangle_d y,
 \end{equation*}
 which contradicts the fact that $x \triangle_d y \in \umin(X \triangle_d \umin(Y))$. Hence $x \triangle_d y \in \umin(X \triangle_d Y)$, and this proves $\umin(X \triangle_d \umin(Y)) \subseteq \umin(X \triangle_d Y)$.

Now let $x \in X$ and $y \in Y$ such that $x \triangle_d y \in \min(X \triangle_d Y)$. Let $y' \in \umin(Y)$ such that $y' \sqsubseteq y$. Then $x \triangle_d y' \sqsubseteq x \triangle_d y$. Since the latter is assumed to be minimal in $X \triangle_d Y$, it follows that $x \triangle_d y' = x \triangle_d y$. In particular, $x \triangle_d y \in X \triangle_d \umin(Y)$. Since $x \triangle_d y$ is minimal in $X \triangle_d Y$, it is certainly minimal in the smaller set $X \triangle_d \umin(Y)$. This proves $\umin(X \triangle_d \umin(Y)) \supseteq \umin(X \triangle_d Y)$.
 \end{proof}

\begin{proof}[Proof of Theorem \ref{thm:aux}]
 We prove this via induction on $v$. From \eqref{eq:probtreebas} it is clear that it holds for BASs. Now suppose $v = \AND(v_1,v_2)$, and that the statement holds for $v_1$ and $v_2$. We can then write \eqref{eq:probtreeand} as
 \begin{equation*}
 \CpU(v) = \umin(H_U(\CpU(v_1) \triangle_d \CpU(v_2)))
 \end{equation*}
 and what we need to prove as
 \begin{equation*}
 \CpU(v) \stackrel{?}{=} \umin(H_U(\recht{EA}_v(\mathcal{A}_v))).
 \end{equation*}
 Using the induction hypothesis and Lemmas \ref{lem:aux1} and \ref{lem:aux2}, we find
 \begin{align*}
 &\CpU(v)\\
 &= \umin H_U[\CpU(v_2) \triangle_d \CpU(v_2)]\\
 &\stackrel{\text{IH}}{=} \umin H_U\left[\umin H_U(\recht{EA}_{v_1}(\mathcal{A}_{v_1})) \triangle_d \umin H_U(\recht{EA}_{v_2}(\mathcal{A}_{v_2}))\right] \\
 &\stackrel{\eqref{eq:auxeq1}}{=}  H_U \umin\left[\umin H_U(\recht{EA}_{v_1}(\mathcal{A}_{v_1})) \triangle_d \umin H_U(\recht{EA}_{v_2}(\mathcal{A}_{v_2}))\right] \\
 &\stackrel{\eqref{eq:auxeq4}}{=} H_U \umin\left[H_U(\recht{EA}_{v_1}(\mathcal{A}_{v_1})) \triangle_d H_U(\recht{EA}_{v_2}(\mathcal{A}_{v_2}))\right] \\
 &\stackrel{\eqref{eq:auxeq1}}{=} \umin H_U\left[H_U(\recht{EA}_{v_1}(\mathcal{A}_{v_1})) \triangle_d H_U(\recht{EA}_{v_2}(\mathcal{A}_{v_2}))\right] \\
 &\stackrel{\eqref{eq:auxeq2}}{=} \umin H_U\left[\recht{EA}_{v_1}(\mathcal{A}_{v_1}) \triangle_d \recht{EA}_{v_2}(\mathcal{A}_{v_2})\right] \\
 &\stackrel{\eqref{eq:aux1}}{=} \umin H_U\left[\recht{EA}_v(\mathcal{A}_v)\right],
 \end{align*}
 which is what needed to be shown. The case that $v = \OR(v_1,v_2)$ is completely analogous.
 \end{proof}
 We are now in a position to prove Theorems \ref{thm:DgC-tree}, \ref{thm:CDPF-tree}, \ref{thm:EDgC-tree} and \ref{thm:CEDPF-tree}. Note that the deterministic scenario can be reduced to the probabilistic scenario, by taking $\pp(v) = 1$ for all $v \in B$; this ensures that $\hdE(\bx) = \hd(\bx)$ for all $\bx$. It also causes the definition of $\CU(v)$ to coincide with that of $\CpU(v)$ for all $v$. Therefore it suffices to prove \ref{thm:EDgC-tree} and \ref{thm:CEDPF-tree}.

\thmEDgCtree*

\begin{proof}
 Let $d_{\recht{E,opt}}$ be the solution to EDgC, i.e., there exists an $\bx_0 = \mathcal{A}$ such that $c_0:= \hc(\bx) \leq U$ and $d_{\recht{E,opt}} = \hdE(\bx)$, and $d_{\recht{E,opt}}$ is maximal under this constraint. Let $x_0 = \recht{EA}(\bx)$; then certainly $x_0 \in H_U(\recht{EA}(\mathcal{A}))$. Let $x' = \vvec{c'\\d'\\p'} \in \umin H_U(\recht{EA}(\mathcal{A}))$ with $x' \sqsubset x$. Then $c' \leq c_0 \leq U$ and $d' \geq d_{\recht{E,opt}}$; since $d_{\recht{E,opt}}$ was assumed to be optimal given $c_0 \leq U$, we conclude that $d = d'$. It follows that 
 \begin{align*}
 d_{\recht{E,opt}} &= \max\left\{d \middle| \exists c,p. \ \vvec{c\\d\\p} \in \umin H_U(\recht{EA}(\mathcal{A}))\right\},\\
 &= \max\left\{d \middle| \exists c,p. \vvec{c\\d\\p} \in \CpU(\recht{R}_T)\right\},
 \end{align*}
 where the second equation follows from Theorem \ref{thm:aux}. This is what was needed to be proven.
 \end{proof}

\thmCEDPFtree*

\begin{proof}
 We claim that $\umin \circ \pi \circ \umin = \umin \circ \pi$ as maps $\Pow(\DP) \rightarrow \Pow(\DD)$. To show this, let $X \subseteq \DP$, and let $x \in \umin(X)$ be such that $\pi(x) \in \umin(\pi(\umin(X)))$. Then $\pi(x) \in \pi(X)$; suppose $\pi(x)$ is not minimal in $\pi(X)$, and there exists an $x' \in X$ such that $\pi(x') \sqsubset \pi(x)$. Let $x'' \in \umin X$ be such that $x'' \sqsubseteq x'$. Since $\pi$ is order-preserving, we have $\pi(x'') \sqsubseteq \pi(x') \sqsubset \pi(x)$. However, this contradicts the fact that $\pi(x)$ is minimal in $\pi(\umin(X))$. Hence our assumption that $\pi(x)$ is not minimal in $\pi(X)$ does not hold, and we can conclude $\umin(\pi(\umin(X))) \subseteq \umin(\pi(X))$.

Now let $x \in X$ be such that $\pi(x) \in \umin(\pi(X))$. Let $x' \in \umin(X)$ be such that $x' \sqsubseteq x$. Since $\pi$ is order-preserving, we find $\pi(x') \sqsubseteq \pi(x)$, but since $\pi(x)$ is minimal, this is an equality; hence $\pi(x) = \pi(x') \in \pi(\umin(X))$. Furthermore, since $\pi(x)$ is minimal in $\pi(X)$, it is certainly minimal in $\pi(\umin(X))$. We conclude $\umin(\pi(\umin(X))) \supseteq \umin(\pi(X))$, which proves the claim. 

Let us now return to the proof of \ref{thm:CEDPF-tree}. Note that $H_{\infty}$ is just the identity on $\Pow(\DP)$, and that $\pi \circ \recht{EA} = \EE$. It follows that
 \begin{align*}
 \umin \pi \mathcal{C}^{\recht{PT}}_{\infty}(\R{T}) &= \umin \pi \umin H_{\infty} \recht{EA}(\mathcal{A}) \\
 &= \umin \pi \recht{EA}(\mathcal{A}) \\
 &= \umin \EE(\FA) \\
 &= \PFE. \qedhere
 \end{align*}
\end{proof}

\end{document}